\documentclass[12pt]{article}

\usepackage{amsthm,amsmath}
\usepackage{subfigure, array, latexsym}
\usepackage{comment}

\usepackage{amssymb}
\usepackage{setspace}
\usepackage{times}
\usepackage{bm}
\usepackage[nomessages]{fp}

\usepackage{natbib}
\RequirePackage{natbib}

\usepackage{latexsym}
\usepackage{graphicx}
\usepackage{todonotes}
\usepackage{subfigure}

\usepackage[dvips]{epsfig}
\usepackage{epstopdf}		
\usepackage{url}
\usepackage{titlesec}

\titleformat{\section}{\large\bfseries}{\thesection.}{1em}{}
\titleformat{\subsection}[runin]{\bfseries}{\thesubsection.}{1em}{}[.]

\def\ps@plain{%
      \let\@oddhead\@empty
      \def\@oddfoot{\normalfont\hfil{\bf\thepage}\hfil}%
      \def\@evenfoot{\normalfont\hfil{\bf\thepage}\hfil}
}

\newtheorem{theorem}{Theorem}
\newtheorem{proposition}{Proposition}

\newtheorem{lemma}{Lemma}

\newtheorem{assumption}{Assumption}
\newtheorem{definition}{Definition}

\theoremstyle{remark}
\newtheorem{remark}{Remark}
\newtheorem{exmp}{Example}
\newtheorem{algorithm}{Algorithm}
\def\qed{\rule{2mm}{2mm}}

\newcommand{\csum}[2]{
  \ifinner  \sum_{#1}^{#2}   \else   \displaystyle\sum\limits_{\substack{#1}}^{\substack{#2}}   \fi
}

\newcommand{\ccup}[2]{
  \ifinner  \bigcup_{#1}^{#2}   \else   \displaystyle\bigcup\limits_{\substack{#1}}^{\substack{#2}}   \fi
}

\newcommand{\ind}[1]{I{\left\{ #1 \right\}}}
\newcommand{\cond}{\;\middle\vert\;}
\newcommand{\expv}[1] {E\left( #1 \right)}
\newcommand{\pr}[1] {\mathrm{Pr}\left( #1 \right)}

\newcommand{\Rd}{R({\mathcal{G}_{d_i}})}
\newcommand{\V}[1]{V(#1)}
\newcommand{\R}[1]{R(#1)}

\newcommand{\cardgdi}{|\mathcal{G}_{d_i}|}
\newcommand{\cardfdi}{|\mathcal{F}_{d_i}|}

\def\firstpage{1}
\newcommand{\pgnum}[1]{\FPeval{\result}{clip(\firstpage+#1)} \result}

\begin{document}

\title{On Procedures Controlling the FDR for Testing Hierarchically Ordered Hypotheses}

\author{
Gavin Lynch \\
Catchpoint Systems, Inc., 228 Park Ave S $\sharp$28080 \\
New York, NY 10003, U.S.A.
\and
Wenge Guo\thanks{The research of
Wenge Guo was supported in part by NSF Grant
DMS-1309162.} \\
Department of Mathematical Sciences \\
New Jersey Institute of Technology \\
Newark, NJ 07102, U.S.A.
}

\date{}
\maketitle

\begin{abstract}
Complex large-scale studies, such as those related to microarray data and fMRI studies, often involve testing multiple hierarchically ordered hypotheses. However, most existing false discovery rate (FDR) controlling procedures do not exploit the inherent hierarchical structure among the tested hypotheses. In this paper, we first present a generalized stepwise procedure which generalizes the usual stepwise procedure to the case where each hypothesis is tested with a different set of critical constants. This procedure is helpful in creating a general framework under which our hierarchical testing procedures are developed. Then, we present several hierarchical testing procedures which control the FDR under various forms of dependence such as positive dependence and block dependence. Our simulation studies show that these proposed methods can be more powerful in some situations than alternative methods such as Yekutieli's hierarchical testing procedure (Yekutieli, \emph{JASA} \textbf{103} (2008) 309-316). Finally, we apply our proposed procedures to a real data set involving abundances of microbes in different ecological environments.
\end{abstract}


\noindent\textsc{Keywords}: {Block dependence, false discovery rate, hierarchical testing,
multiple testing, PRDS property, $p$-values, stepwise procedure.}

\vskip 30pt

\section*{Notation Index}
The following summarizes commonly used notation and lists where each symbol is found.

\noindent
    \begin{tabular}{cp{11cm}cc}
    \textbf{Symbol}      & \textbf{Description}                                          & Section & Page \\
    $\mathcal{M}, m$     & The set of tested hypotheses $\{H_1, \dots, H_m\}$ and its cardinality. & 2 & \pageref{NOTATION_M}    \\
    $\mathcal{M}_i, m_i$ & The set of descendant hypotheses of $H_i$ and its cardinality.  & 2       & \pageref{NOTATION_MI}    \\
    $\mathcal{D}_i, d_i$ & The set of ancestor hypotheses of $H_i$ and its cardinality, also referred to as its depth. & 2       & \pageref{NOTATION_DI}    \\
    $T(\cdot)$                  & A function that takes an index of a hypothesis and returns the index of its parent hypothesis. & 2       & \pageref{NOTATION_T}    \\
    $\mathcal{F}_d$    & The set of hypotheses with depth $d$, $\mathcal{F}_d = \{H_i : d_i = d\}$.             & 2       & \pageref{NOTATION_F}    \\
    $\mathcal{G}_d$    & The union of $\mathcal{F}_1, \dots, \mathcal{F}_d$.       & A.2       & \pageref{NOTATION_G} \\
    $D$                & The maximum depth of the hypotheses $\mathcal{G}_d \subseteq \mathcal{M}$ so that $\mathcal{G}_D = \mathcal{M}$.                                 & 2       & \pgnum{3}    \\
    $\ell$             & The total number of leaf hypotheses.                          & 2       & \pageref{NOTATION_L}    \\
    $\ell_i$           & The number of leaf hypotheses in set $\mathcal{M}_i$.         & 2       & \pageref{NOTATION_LI}    \\
    $R(\mathcal{A}), R$   & The number of rejected hypotheses belonging to any set $\mathcal{A}$ and $R = R(\mathcal{M})$. & 2    & \pageref{NOTATION_R}   \\
    $V(\mathcal{A}), V$   & The number of falsely rejected hypotheses belonging to any set $\mathcal{A}$ and $V = V(\mathcal{M})$. & 2    & \pageref{NOTATION_V}    \\
    $\alpha_i(\cdot)$         & The critical function for testing the $i^{th}$ hypothesis $H_i$.    & 2    & \pageref{NOTATION_CRIT_FUNC}    \\
    \end{tabular}
\pagebreak

\section{Introduction}
In many problems involving the testing of multiple hypotheses, the hypotheses have an intrinsic, hierarchical structure such as a tree-like or graphical structure. These hierarchical structures often arise in multiple testing problems involving clinical trials \citep{Mehrotra_2004, Dmitrienko_2007, Huque_2008}, genomics research \citep{Yekutieli_2006, Goeman_2008, Heller_2009, Guo_2010} and fMRI studies \citep{Benjamini_2007}. In general, hierarchical testing typically occurs while testing hierarchically structured hypotheses where, upon the rejection of one hypothesis, followup hypotheses are to be tested. For instance, \cite{Heller_2009} introduced a hierarchical testing approach for analyzing microarray data where individual genes were grouped into gene sets. The gene sets were tested and upon successfully rejecting a gene set, the associated individual genes were tested. \cite{Guo_2010} and \cite{Mehrotra_2004} used a similar hierarchical testing approach for time-course microarray data and clinical safety data, respectively. \cite{Benjamini_2007} used a hierarchical testing approach to study fMRI data where the brain was divided into brain regions and each brain region was tested for significance. If a brain region was significant, the voxels within the brain region were tested. In addition, \cite{Meinshausen_2008} introduced a hierarchical testing method for addressing the problem of variable selection in multiple linear regression models.

In the field of multiple testing, the problem of controlling the familywise error rate (FWER) for testing hierarchically ordered hypotheses has received considerable attention \citep{Dmitrienko_2006, Dmitrienko_2007, Goeman_2008, Huque_2008, Meinshausen_2008, Brechenmacher_2011, Goeman_2012}; however, the FWER control can be too conservative for large-scale multiple testing. There has been very few work towards developing general methods for testing hierarchically ordered hypotheses that control the false discovery rate (FDR), even though the FDR is a more appropriate error measure for large scale multiple testing. To our knowledge, only \cite{Yekutieli_2008} has provided a general method for testing hierarchically ordered hypotheses that is specifically intended for controlling the FDR. Yekutieli's procedure, which is based on the Benjamini-Hochberg (BH) procedure \citep{Benjamini_1995}, is only shown to control the FDR under independence. Some of the aforementioned procedures \citep{Mehrotra_2004, Benjamini_2007, Heller_2009, Guo_2010} can only be applied to special hierarchies consisting of only two layers.

In this paper, we propose new FDR controlling methods for testing hierarchically ordered hypotheses under various dependencies. Our approach towards controlling the FDR for testing hierarchically ordered hypotheses is different from that of Yekutieli's. First, to assist in the development of our hierarchical testing procedures, we introduce a new concept of generalized stepwise procedure, which generalizes the usual stepup, stepdown, and stepup-down procedures to the case where each hypothesis is tested with a different set of critical constants. The hypotheses are organized into different families according to their depth in the hierarchical structure. The formed families are sequentially tested by using the generalized stepwise procedures for which the corresponding critical constants take into account of the testing outcomes of higher-ranked families. Based on this approach, we were able to develop several new hierarchical testing procedures which control the FDR under various dependence structures including positive dependence and block dependence. To our knowledge, the procedures are the first procedures developed for testing hierarchically ordered hypotheses with proven control of the FDR under dependence structures other than independence. Furthermore, our simulation study shows that these procedure are quite powerful. The most powerful procedure, which we prove controls the FDR under positive block dependence, significantly outperforms Yekutieli's procedure in terms of power even though Yekutieli's procedure is only shown to control the FDR under independence, which is a special case of positive block dependence.

Another interesting finding of this research is that when the hierarchy takes on some special configurations, our procedures reduce to the existing FDR controlling procedures. For example, when there is no hierarchical structure, our proposed procedures reduce to the BH procedure and the Benjamini-Yekutieli (BY) procedure \citep{Benjamini_2001}. When the hierarchy takes on a fixed sequence structure, our procedures are equivalent to the fixed sequence procedures in \cite{Lynch_2014}. This shows that our procedures are the combination of stepwise and fixed sequence methods.

The rest of this paper is outlined as follows. In Section 2, we provide relevant notation and definitions that will be used throughout this paper. Section 3 presents our proposed generalized stepwise procedure. Section 4 presents our new hierarchical testing procedures with proven control of the FDR under various dependencies. Sections 5 and 6 present a simulation study and real data analysis where we compare our procedures with Yekutieli's procedure. Finally, Section 7 provides some brief discussions.

\begin{figure}
\centering
\mbox{\subfigure[]{\includegraphics[scale=.25, trim=0 175 0 100]{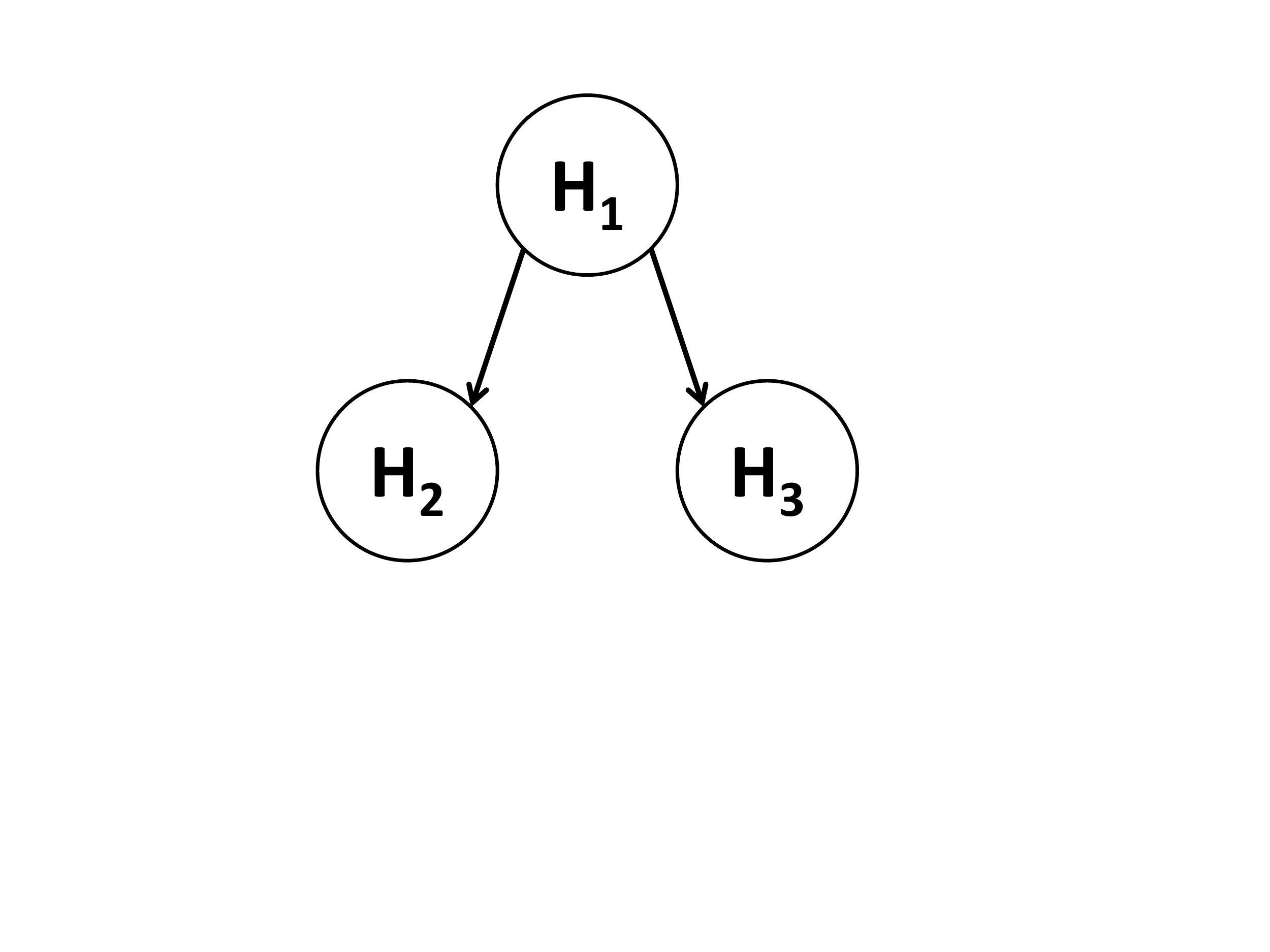}}\quad
\subfigure[]{\includegraphics[scale=.25, trim=0 50 0 100]{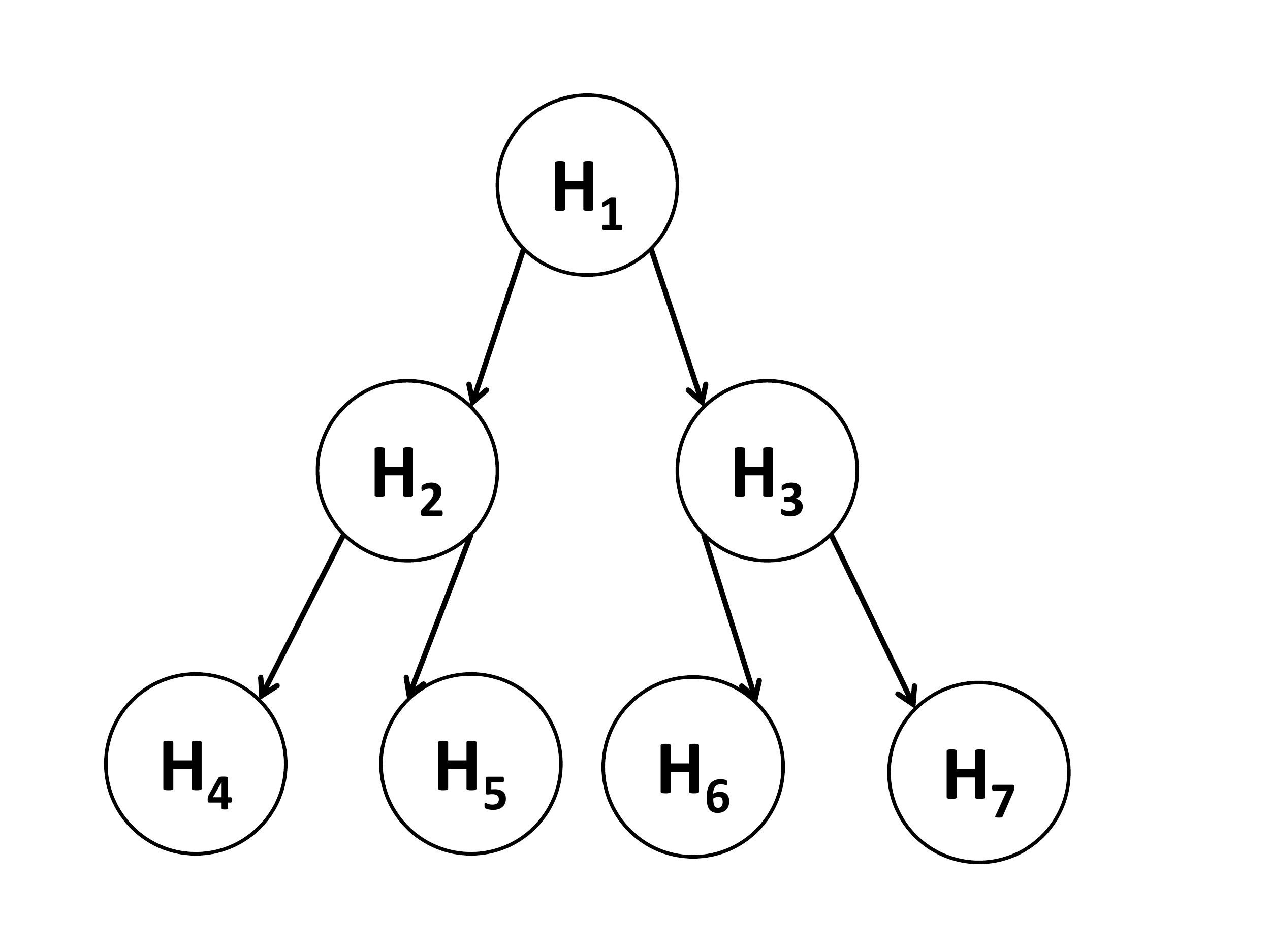}}}
\caption{(a) An example of a hierarchical structure with 3 hypotheses for which $H_2$ and $H_3$ are only tested if $H_1$ is rejected. (b) An example of a hierarchical structure with 7 hypotheses for which $H_2$ and $H_3$ are only tested if $H_1$ is rejected, $H_4$ and $H_5$ are only tested if $H_2$ is rejected, and $H_6$ and $H_7$ are only tested if $H_3$ is rejected.}
\label{IMG_EXAMPLE}
\end{figure}

\section{Preliminaries}

Suppose there are $m$ hypotheses $H_1, \dots, H_m$ to be tested that are organized hierarchically in a tree-like structure where each hypothesis can have several child hypotheses but at most one parent hypothesis. Let $\mathcal{M} = \{H_1, \dots, H_m\}$\label{NOTATION_M} be the set of the $m$ tested hypotheses. Let $T: \{0, \dots, m\} \rightarrow \{0, \dots, m\}$\label{NOTATION_T} be a function that takes an index of a hypothesis and returns the index of the parent hypothesis with $T(0) = 0$. That is, if $H_i$ has a parent hypothesis, its parent hypothesis is $H_{T(i)}$; otherwise $H_i$ does not have a parent hypothesis and $T(i) = 0$.  Define $T^0(i) = i$ and $T^k(i) = T(T^{k-1}(i))$ for any positive integer $k$. Let $\mathcal{D}_i = \{H_j : T^k(i) = j \text{ for } k = 0, \dots, m \}$\label{NOTATION_DI} so that $\mathcal{D}_i$ is the set of all ancestor hypotheses of $H_i$, which includes $H_i$. Let $d_i$ be the cardinality of $\mathcal{D}_i$, $d_i = |\mathcal{D}_i|$.
The depth of $H_i$ in the hierarchy is defined as $d_i$. Let $D$ be the maximum depth of the $m$ hypotheses to be tested. If $d_i = 1$, then $H_i$ does not have a parent hypothesis. Let $\mathcal{M}_i = \{H_j : T^k(j) = i \text{ for } k = 0, \dots, m \}$\label{NOTATION_MI} so that $\mathcal{M}_i$ is the set of all descendant hypotheses of $H_i$, which also includes $H_i$. We will refer to the hypotheses in set $\mathcal{M}_i$ as the subtree under $H_i$. Let $m_i$ be the cardinality of $\mathcal{M}_i$, $m_i = |\mathcal{M}_i|$. If $m_i = 1$, then $H_i$ has no children and it is referred to as a leaf hypothesis. We denote the number of leaf hypotheses in the whole hierarchy by $\ell$\label{NOTATION_L} and the number of leaf hypotheses in the subtree under $H_i$ by $\ell_i$. Formally, $\ell = \csum{H_j \in \mathcal{M}}{}\ind{m_j = 1}$ and $\ell_i = \csum{H_j \in \mathcal{M}_i}{}\ind{m_j = 1}$\label{NOTATION_LI}. Our procedures introduced in Section 4 group the hypotheses into $D$ families by depth where family $d$ contains all hypotheses with depth $d$, that is, $\mathcal{F}_d = \{H_i \in \mathcal{M} : d_i = d\}$\label{NOTATION_F}. For example, in Figure \ref{IMG_EXAMPLE}(a), $T(2) = T(3) = 1$ and $H_2$ and $H_3$ are leaf hypotheses. In Figure \ref{IMG_EXAMPLE}(b), $T(6) = T(7) = 3, \mathcal{D}_6 = \{H_1, H_3, H_6\}, \mathcal{M}_2 = \{H_2, H_4, H_5\}$, and $\mathcal{F}_3 = \{H_4, H_5, H_6, H_7\}$.

The hypotheses in the hierarchical structure are tested hierarchically by a testing procedure based on their corresponding $p$-values $P_1, \dots, P_m$. By hierarchical testing, we mean a hypothesis is only tested if its parent hypothesis has been rejected or it does not have a parent hypothesis. For any set $\mathcal{A} \subseteq \mathcal{M}$, define $\R{\mathcal{A}}$\label{NOTATION_R} and $\V{\mathcal{A}}$\label{NOTATION_V} to be the number of rejected hypotheses and falsely rejected hypotheses among the hypotheses in set $\mathcal{A}$, respectively. For example, $\R{\mathcal{M}}$ and $\V{\mathcal{M}}$ are the number of rejected hypotheses and falsely rejected hypotheses among all the $m$ tested hypotheses, respectively, and $\R{\mathcal{M}_i}$ and $\V{\mathcal{M}_i}$ are number of rejected hypotheses and falsely rejected hypotheses among the hypotheses in the subtree $\mathcal{M}_i$, respectively. For simplicity, often we will use $R$ and $V$ to denote $\R{\mathcal{M}}$ and $\V{\mathcal{M}}$, respectively. The FWER of this procedure is defined as $\pr{V > 0}$. The FDR of this procedure is defined as FDR = $\expv{V/R}$, where we use the convention that $V/R = 0$ when $R = 0$.
In addition, we will always use $|\mathcal{A}|$ to denote the cardinality of any set $\mathcal{A}$ throughout the paper.

Most existing multiple testing procedures are stepwise methods which are based on the ordered $p$-values $P_{(1)} \le \cdots \le P_{(m)}$ with corresponding hypotheses $H_{(1)}, \dots, H_{(m)}$. Typically the rejection thresholds of a stepwise procedure are based on a sequence of non-decreasing critical constants but in this paper, for convenience, we will instead test the hypotheses using a non-decreasing, non-negative function $\alpha_0: \{0, \dots, m+1\} \rightarrow \mathbb{R}$\label{NOTATION_CRIT_FUNC} called a critical function where $\alpha_0(0) = 0$. For example, the critical function of the BH procedure is $\alpha_0(r) = r\alpha/m$. A stepwise procedure first determines the number of rejections $R$ based on the critical function, then for each $i = 1, \dots, m$, it rejects $H_i$ if $P_i \le \alpha_0(R)$ and accepts $H_i$ if $P_i > \alpha_0(R)$. With $P_{(0)} \equiv 0$ and $P_{(m+1)} \equiv \infty$, a stepup procedure sets $R = \max\{0 \le r \le m: P_{(r)} \le \alpha_0(r)\}$. A stepdown procedure sets $R = \min\{1 \le r \le m+1: P_{(r)} > \alpha_0(r)\} - 1$. Finally, a stepup-down procedure of order $k$, which generalizes stepup and stepdown procedures, sets $R = \max\{0 \le r \le k-1: P_{(r)} \le \alpha_0(r)\}$ if $P_{(k)} > \alpha_0(k)$ and $R = \min\{k+1 \le r \le m+1: P_{(r)} > \alpha_0(r)\}-1$ if $P_{(k)} \le \alpha_0(k)$. When $k = m$, the stepup-down procedure reduces to the stepup proceudre and when $k = 1$, it reduces to the stepdown procedure. It should be noted that the event $\{P_{(r)} \le \alpha_0(r)\}$ is equivalent to the event $\{r \le \sum_{i=1}^{m}\ind{P_i \le \alpha_0(r)}\}$. Thus, the number of rejections can also be expressed by
\begin{equation}
R = \max\left\{0 \le r \le m: r \le \sum_{i = 1}^{m}\ind{P_i \le \alpha_0(r)}\right\} \label{EQN_STEPUP}
\end{equation}
for the stepup procedure,
\begin{equation}
R = \min\left\{1 \le r \le m+1 : r > \sum_{i = 1}^{m}\ind{P_i \le \alpha_0(r)}\right\} - 1 \label{EQN_STEPDOWN}
\end{equation}
for the stepdown procedure, and
\begin{eqnarray}
R =
\begin{cases}
\max\left\{0 \le r \le k-1: r \le \sum_{i = 1}^{m}\ind{P_i \le \alpha_0(r)}\right\} &\text{ if } k > \sum_{i = 1}^{m}\ind{P_i \le \alpha_0(k)} \\
\min\left\{k+1 \le r \le m+1 : r > \sum_{i = 1}^{m}\ind{P_i \le \alpha_0(r)}\right\} - 1 &\text{ if } k \le \sum_{i = 1}^{m}\ind{P_i \le \alpha_0(k)}
\end{cases} \label{EQN_STEPUPDOWN}
\end{eqnarray}
for the stepup-down procedure of order $k$.
Refer to \cite{Tamhane_1998} and \cite{Sarkar_2002} for further discussion on stepup-down procedure.

Throughout this paper we make use of the following basic assumption regarding marginal $p$-values:  for any $p$-value $P_i$ with $H_i$ being true,
\begin{equation}
\pr{P_i \le p} \le p ~~\text{ for any } 0 \le p \le 1. \label{EQN_UNIFORM}
\end{equation}
We consider several types of joint dependence throughout this paper: arbitrary dependence, positive dependence, and block dependence. Under arbitrary dependence, the $p$-values are not known to have any specific type of dependence structure. Positive dependence and block dependence are characterized by the following assumptions.

\begin{assumption} \label{ASM_POS_DEPENDENCE} \emph{Positive Dependence Assumption} \\
For any coordinatewise non-decreasing function of the $p$-values $\psi$,
\begin{equation}
\expv{\psi(P_1, \dots, P_m) \cond P_i \le p} \text{ is non-decreasing in $p$ for each $p$-value $P_i$ such that $H_i$ is true}.
\end{equation}
\end{assumption}

\begin{assumption} \label{ASM_BLOCK_DEPENDENCE} \emph{Block Dependence Assumption} \\
For each $d = 1, \dots, D$, the $p$-values corresponding to the hypotheses in $\mathcal{F}_d$ are independent of the $p$-values corresponding to the hypotheses not in $\mathcal{F}_d$.
\end{assumption}

Assumption \ref{ASM_POS_DEPENDENCE} is slightly more relaxed than the condition of positive regression dependence on a subset (PRDS) introduced in \cite{Benjamini_2001}. Assumption \ref{ASM_BLOCK_DEPENDENCE} only characterizes the joint dependence of the $p$-values across families but does not describe the joint dependence within families.

\section{Generalized Stepwise Procedure}

In order to present our hierarchical testing procedures in the next section, in this section, we present a new type of procedure termed as generalized stepwise procedure, including generalized stepup, stepdown, and stepup-down procedures, which generalizes the usual stepup, stepdown, and stepup-down procedures. In a non-hierarchical multiple testing problem where a stepwise procedure is used to test the hypotheses, the tested hypotheses often have the same importance and thus, it is natural to test those hypotheses with the same critical function, as shown in (\ref{EQN_STEPUPDOWN}). However, when the hypotheses have a hierarchical structure, the importance of a hypothesis depends on where it is located in the hierarchy. Hence, for a desired procedure,
each hypothesis should be tested with a different critical function that reflects its importance, and so we generalize the usual stepwise procedure as follows.

Given $m$ non-decreasing critical functions $\alpha_i(r), i = 1, \ldots, m$, our proposed generalized stepwise procedure rejects $H_i$ if $P_i \le \alpha_i(R)$ for each $i = 1, \dots, m$ where $R$ is determined as follows. For the generalized stepup procedure,
\begin{equation}
R = \max\left\{0 \le r \le m: r \le \sum_{i = 1}^{m}\ind{P_i \le \alpha_i(r)}\right\}, \label{EQN_GENERALIZED_SU}
\end{equation}
for the generalized stepdown procedure,
\begin{equation}
R = \min\left\{1 \le r \le m+1 : r > \sum_{i = 1}^{m}\ind{P_i \le \alpha_i(r)}\right\} - 1, \label{EQN_GENERALIZED_SD}
\end{equation}
and for the generalized stepup-down procedure of order $k$,
\begin{eqnarray}
R =
\begin{cases}
\max\left\{0 \le r \le k-1: r \le \sum_{i = 1}^{m}\ind{P_i \le \alpha_i(r)}\right\} &\text{ if } k > \sum_{i = 1}^{m}\ind{P_i \le \alpha_i(k)} \\
\min\left\{k+1 \le r \le m+1 : r > \sum_{i = 1}^{m}\ind{P_i \le \alpha_i(r)}\right\} - 1 &\text{ if } k \le \sum_{i = 1}^{m}\ind{P_i \le \alpha_i(k)}.
\end{cases} \label{EQN_GENERALIZED_SUSD}
\end{eqnarray}

It is easy to see that when $\alpha_i(r) = \alpha_0(r)$ for each $i = 1, \dots, m$, (\ref{EQN_GENERALIZED_SU}), (\ref{EQN_GENERALIZED_SD}), and (\ref{EQN_GENERALIZED_SUSD}) reduce to (\ref{EQN_STEPUP}), (\ref{EQN_STEPDOWN}), and (\ref{EQN_STEPUPDOWN}), respectively. Thus, the generalized stepup, stepdown, and stepup-down procedures reduce to the usual stepwise procedures, respectively. It should be noted that when $k = m$, (\ref{EQN_GENERALIZED_SUSD}) reduces to (\ref{EQN_GENERALIZED_SU}) and when $k = 1$, (\ref{EQN_GENERALIZED_SUSD}) reduces to (\ref{EQN_GENERALIZED_SD}).

The generalized stepwise procedure is fairly general and we present two examples to show its broad applicability.
\begin{exmp} \label{EXMP_WEIGHTED_PROCEDURE}
Consider a weighted multiple testing problem where $H_i$ has corresponding weight $w_i, i = 1, \dots, m$. A weighted stepwise procedure with the critical function $\alpha_0(r)$ tests $H_i$ based on weight-adjusted $p$-values $P_i/w_i$ instead of $P_i$. This is equivalent to a generalized stepwise procedure with the critical functions $\alpha_i(r) = w_i\alpha_0(r), i = 1, \ldots, m$ so that the weighted stepwise procedure can be regarded as a special case of the generalized stepwise procedure.
\end{exmp}
\begin{exmp} \label{EXMP_FIXED_SEQ_PROCEDURE}
Fixed sequence procedures assume the testing order of the hypotheses has been specified a-priori and that $H_i$ is not tested unless $H_1, \dots, H_{i-1}$ have all been rejected. \cite{Lynch_2014} showed that the fixed sequence procedure that rejects $H_i$ when $P_i \le m\alpha/(m-i+1)$ controls the FDR at level $\alpha$ under arbitrary dependence. This procedure is a special case of the generalized stepdown procedure with critical functions $\alpha_i(r) = \ind{r \ge i}m\alpha/(m-r+1)$. Other fixed sequence procedures can be defined similarly.
\end{exmp}

From (\ref{EQN_GENERALIZED_SUSD}), it can be seen that many of the familiar properties of stepwise procedures also hold for the generalized stepwise procedure. For example, the number of rejections $R$ is a coordinatewise non-increasing function of the $p$-values and $R$ is a non-decreasing function of $k$ (i.e. a stepup-down procedure of order $k$ rejects more hypotheses than a stepup-down procedure of order $k-1$). The most important property is a self-consistency property which allows us to express $R$ as
\begin{equation}
R = \csum{i=1}{m}\ind{P_i \le \alpha_i(R)}. \label{EQN_EXACT_SELF_CONSISTENT}
\end{equation}
(Blanchard and Roquain (2008) discussed a weaker self-consistency condition for the usual stepwise procedure with the critical function $\alpha_0(r)$, which is the inequality $R \le \csum{i=1}{m}\ind{P_i \le \alpha_0(R)}$). This property ensures that $R$ as determined in (\ref{EQN_GENERALIZED_SUSD}) is indeed the number of rejections by the generalized stepwise procedure. Thus, the event $\{H_i \text{ is rejected}\}$ can be expressed as $\{P_i \le \alpha_0(R)\}$ with $R$ being the number of rejections. To see why this property holds, let us define $\psi(r) = \sum_{i=1}^{m}\ind{P_i \le \alpha_i(r)}$. When $k > \psi(k)$, then $R = \max\{0 \le r \le k-1: r \le \psi(r)\}$ and if $k \le \psi(k)$, then $R+1 = \min\{k+1 \le r \le m+1: r > \psi(r)\}$. In either case, it is easy to see that $R \le \psi(R)$ and $R+1 > \psi(R+1)$. The fact that $\psi(R+1) < R+1$ implies $\psi(R+1) \le R$. Thus, $R = \psi(R)$ since $R \le \psi(R) \le \psi(R+1) \le R$.

To conclude this section, we present an efficient algorithm for finding the number of rejections by the generalized stepwise procedure. The algorithm is particularly useful when the number of hypotheses is very large.
\vspace{5pt}
\begin{algorithm} \label{PROC_FIND_R}
Given a positive integer $1 \le k \le m$ and critical functions $\alpha_i(\cdot), i = 1, \dots, m$, define $\psi(r) =\csum{i = 1}{m}\ind{P_i \le \alpha_i(r)}$.
\begin{enumerate}
\item Let $t = 1$ and $r_t = k$.
\item If $r_t > \psi(r_t)$, then
\begin{enumerate}
	\item Increase $t$ by 1 and set $r_{t} = \psi(r_{t-1})$.
	\item If $r_{t} \le \psi(r_{t})$, then let $R = r_{t}$ and stop; otherwise, if $r_{t} > \psi(r_{t})$, repeat step 2(a).
\end{enumerate}
\item Otherwise, if $r_t \le \psi(r_t)$, then
\begin{enumerate}
	\item Increase $t$ by 1 and set $r_{t} = \psi(r_{t-1})+1$.
	\item If $r_{t} > \psi(r_{t})$, then let $R = r_{t}-1$ and stop; otherwise, if $r_{t} \le \psi(r_{t})$, repeat step 3(a).
\end{enumerate}
\end{enumerate}
\end{algorithm}

\vspace{10pt}
\begin{proposition} \label{PROP_ALGORITHM}
The value of $R$ in (\ref{EQN_GENERALIZED_SUSD}) can be solved by algorithm \ref{PROC_FIND_R}.
\end{proposition}
\begin{proof}
The proof is in the Appendix.
\end{proof}

\section{Hierarchical FDR Control}

In this section, we describe our procedure to test hierarchically ordered hypotheses. The tested hypotheses are arranged into $D$ families, $\mathcal{F}_1, \dots, \mathcal{F}_D$, where $\mathcal{F}_d$ is the family of hypotheses with depth $d$. Given $m$ non-decreasing critical functions $\alpha_i(r), i=1, \ldots, m$, the hypotheses are tested as follows.

\begin{definition} \label{DEF_HIER} \emph{General Hierarchical Testing Procedure}
\begin{enumerate}
\item Test $\mathcal{F}_1$ by using the generalized stepup procedure with critical functions $\alpha_i(r), H_i \in \mathcal{F}_1$. Let $\mathcal{S}_1$ be the set of rejected hypotheses and $R(\mathcal{F}_1)$ be the number of rejected hypotheses in $\mathcal{F}_1$. Test $\mathcal{F}_2$.
\item To test $\mathcal{F}_d$, use the generalized stepup procedure with critical functions $$\alpha^*_i (r) = \ind{H_{T(i)} \text{ is rejected}} \allowbreak \alpha_i \left(r + \sum_{j=1}^{d-1}R(\mathcal{F}_{j}) \right), ~~H_i \in \mathcal{F}_d.$$ Let $\mathcal{S}_d$ be the set of rejected hypotheses and $R(\mathcal{F}_d)$ be the number of rejected hypotheses in $\mathcal{F}_d$. Test $\mathcal{F}_{d+1}$.
\item The set of rejected hypotheses is $\bigcup_{d=1}^{D}\mathcal{S}_d$ and the total number of rejections is $R = \sum_{d = 1}^{D}R(\mathcal{F}_d)$.
\end{enumerate}
\end{definition}

The above procedure is termed as a hierarchical testing procedure since the procedure will accept any hypothesis whose parent hypothesis has been accepted. This can be seen in the construction of the critical functions in step 2 where $\alpha_i^*(r) = 0$ if $H_i$'s parent hypothesis $H_{T(i)}$ has not been rejected so that $H_i$ cannot be rejected. It should be noted that the parents of all hypotheses in $\mathcal{F}_d$ are in $\mathcal{F}_{d-1}$, which is tested before testing $\mathcal{F}_d$. Hence, for each $H_i \in \mathcal{F}_d$, the event $\{H_{T(i)} \text{ is rejected}\}$ is observed by the time $\mathcal{F}_d$ is tested.

\vspace{5pt}
\begin{remark}
In Definition 1, when all the hypotheses in $\mathcal{F}_d$ have the same critical functions and every $H_i \in \mathcal{F}_d$ can be tested (i.e. $H_{T(i)}$ is rejected), the generalized stepup procedure used for testing $\mathcal{F}_d$ reduces to the usual stepup procedure. However, our critical functions for testing hierarchically ordered hypotheses, which are presented in the next subsections, are not the same and depend on where the hypothesis is located in the hierarchy. Furthermore, since the hypotheses in $\mathcal{F}_d$ may not have the same parent, the parent hypotheses $H_{T(i)}$ could be rejected for some, but not all, of $H_i \in \mathcal{F}_d$. Hence, only in an uncommon case does the generalized stepup procedure reduce to the usual stepup procedure for testing $\mathcal{F}_d$.
\end{remark}

Following (\ref{EQN_EXACT_SELF_CONSISTENT}), the hierarchical testing procedure has the following self-consistency property in each family $\mathcal{F}_d$,
\[
\R{\mathcal{F}_d} = \csum{H_i \in \mathcal{F}_d}{}\ind{P_i \le \alpha_i^*(\R{\mathcal{F}_d})}, d= 1, \dots, D,
\]
where $\alpha^*_i (r) = \ind{H_{T(i)}\text{ is rejected}}\alpha_i(r + \sum_{j=1}^{d-1}R(\mathcal{F}_{j}))$. Hence, the event $\{H_i \text{ is rejected}\}$ is equivalent to the event $ \{H_{T(i)}\text{ is rejected}, P_i \le \alpha_i(\sum_{j=1}^{d_i}R(\mathcal{F}_j)) \}$, where $\sum_{j=1}^{d_i}R(\mathcal{F}_j)$ is the number of rejections in the first $d_i$ families, $\mathcal{F}_1, \ldots, \mathcal{F}_{d_i}$. This property will be useful to prove the FDR control of our procedures.

Now that we have defined our hierarchical testing procedure, we will consider various dependence structures, such as positive dependence, arbitrary dependence, and block dependence, and develop newer hierarchical testing procedures which control the FDR under these dependence structures. The proofs of all the theorems in this section are in the appendix.

\subsection{Procedure under Positive Dependence}
We first consider positive dependence structure. Positive dependence has received much attention in multiple testing due to the fact that several popular multiple testing procedures have been developed under this type of dependence (see \cite{Sarkar_1998, Benjamini_2001, Sarkar_2002, Guo_2012}). Our procedure under positive dependence is as follows.

\begin{theorem} \label{THM_HIER_FDR_POS_DEP} FDR Control under Positive Dependence\\
Under Assumption \ref{ASM_POS_DEPENDENCE}, the hierarchical testing procedure with critical functions
\[
\alpha_i(r) = \frac{\ell_i \alpha}{\ell}\frac{m_i + r-1}{m_i}, ~~i = 1, \ldots, m,
\]
strongly controls the FDR at level $\alpha$.
\end{theorem}

Consider the special case when there is no hierarchical ordering (i.e. $\ell = m$ and $\ell_i = m_i = 1$) so that the hypotheses do not have any pre-defined structure and the problem reduces to a non-hierarchical multiple testing problem. We will refer to this configuration as the non-hierarchical configuration. Under this configuration, all the hypotheses belong to the same family, $\mathcal{F}_1$, so that the hierarchical testing procedure reduces to the usual stepup procedure which can be further reduced to a normal stepup procedure since all the critical functions are equal to $r\alpha/m$. Hence, the procedure reduces to the stepup procedure with critical function $r\alpha/m$, which is the BH procedure. Thus, our result generalizes the BH to the testing of hierarchically ordered hypotheses.

Now, we consider another special case where each family $\mathcal{F}_i$ has exactly one hypothesis $H_i, i=1, \ldots, m$. Thus, the tested hypotheses $H_1, \ldots, H_m$ are pre-ordered, $\ell = 1$, and $m_i = m-i+1$.
We will refer to this configuration as the fixed sequence configuration. Under this configuration, the hierarchical testing procedure reduces to the fixed sequence method introduced in \cite{Lynch_2014}, where hypothesis $H_i$ is rejected if, and only if, hypotheses $H_1, \dots, H_{i-1}$ have all been rejected and $P_i \le m\alpha/(m-i+1)$. This method is shown to control the FDR at level $\alpha$ under arbitrary dependence. Thus, our result also generalizes the fixed sequence procedure to the testing of hierarchically ordered hypotheses.

Remarkably, our result has connected two opposing testing methods: the testing of non-ordered hypotheses (through the BH procedure) and the testing of fully ordered hypotheses (through the fixed sequence procedure).

Finally, we consider a third configuration which we call the binary tree configuration. This configuration is helpful for evaluating the critical functions in the hierarchical setting and it is defined as follows. There is one hypothesis in $\mathcal{F}_1$ and each hypothesis has two child hypotheses except for the leaf hypotheses in $\mathcal{F}_D$. Hence, $\ell = 2^{D-1}$ and $m = 2^D-1$. For each $d = 1, \dots, D$, there are $2^{d-1}$ hypotheses in $\mathcal{F}_d$ and for each $H_i \in \mathcal{F}_d$, $\ell_i = 2^{D-d}$ and $m_i = 2^{D-d+1}-1$. Under this configuration, the critical functions of Theorem \ref{THM_HIER_FDR_POS_DEP} are, after simplification,
\begin{equation}
\alpha_i(r) = \frac{\alpha}{2^{d-1}}\left(1+\frac{r-1}{2^{D-d+1}-1}\right), ~~H_i \in \mathcal{F}_d, d = 1, \dots, D. \label{EQN_BINARY_TREE_POS}
\end{equation}

Compared to Meinshausen's FWER controlling hierarchical testing procedure, which is equivalent to the hierarchical testing procedure with critical functions $\ell_i\alpha/\ell, i = 1, \dots, m$, the critical functions of Theorem \ref{THM_HIER_FDR_POS_DEP} are $(m_i+r-1)/m_i$ times larger for $H_i$. In the binary tree configuration, Meinshausen's critical function for $H_i$ in $\mathcal{F}_d$ is $\alpha_i(r) = \alpha/2^{d-1}$, which is $(1 + (r-1)/(2^{D-d+1}-1))$ times smaller than the critical function in Theorem \ref{THM_HIER_FDR_POS_DEP}. Table \ref{TBL_POS_DEP} lists the critical functions of Theorem \ref{THM_HIER_FDR_POS_DEP} and Meinshausen's procedure for testing the hypotheses in Figure \ref{IMG_EXAMPLE}(b) which has the binary tree configuration. For family $d$, only the values of $r$ between $d$ and $\sum_{j=1}^{d}|\mathcal{F}_j|$ are listed in Table \ref{TBL_POS_DEP} due to the fact that if a hypothesis in family $d$ is rejected, then all $d$ of its ancestor hypotheses including the hypothesis itself are rejected so that $d \le \sum_{j=1}^{d}R(\mathcal{F}_j) \le \sum_{j=1}^{d}|\mathcal{F}_j|$.

\begin{table}
    \begin{center}
    \caption{A comparison of critical functions for the procedure in Theorem \ref{THM_HIER_FDR_POS_DEP} and Meinshausen's procedure when testing the hypotheses in Figure \ref{IMG_EXAMPLE}(b).}
    \vskip 5pt
    \begin{tabular}{|c|c||c|c|c|c|c|c|c||c|}
    	\hline
        ~ & ~ & \multicolumn{7}{|c||}{Theorem \ref{THM_HIER_FDR_POS_DEP}} & Meinshausen  \\ \noalign{\hrule height 2pt}
        ~ & ~ & $r=1$ & $r=2$ & $r=3$ & $r=4$ & $r=5$ & $r=6$ & $r=7$ &   \\ \noalign{\hrule height 2pt}
        Family 1 & $\alpha_i(r)$ & $\alpha$ & -   &  -  &  -  & -   & -   &  - & $\alpha$ \\ \hline
        Family 2 & $\alpha_i(r)$ & -   & $2\alpha/3$   &  $5\alpha/6$  &  -  & -   & -   & - & $\alpha/2$ \\ \hline
        Family 3 & $\alpha_i(r)$ & -   & - & $ 3\alpha/4$  &  $\alpha$  &  $5\alpha/4$  & $3\alpha/2$   & $7\alpha/4$ & $\alpha/4$ \\ \hline
    \end{tabular} \label{TBL_POS_DEP}
    \end{center}
\end{table}

\normalsize

\begin{remark}
It should be noted that the hierarchical testing procedure relies on the generalized stepup procedure to test each family; however, our proof of the FDR control for the procedure in Theorem \ref{THM_HIER_FDR_POS_DEP} (and our proof of FDR control for the remaining procedures in this section) still holds if the generalized stepup-down procedure of any arbitrary order (including the generalized stepdown procedure) is used to test each family. Nevertheless, in practice we are generally trying to maximize the number of rejections subject to the FDR control. Since with the same critical functions, the generalized stepup procedure is more powerful than the corresponding generalized stepup-down and generalized stepdown procedures, we opted to use the generalized stepup procedure to test each family.
\end{remark}

\subsection{Procedure under Arbitrary Dependence}
In this subsection we introduce a FDR controlling hierarchical testing procedure under arbitrary dependence. Since arbitrary dependence is a more general type of joint dependence than positive dependence, it follows that the procedure under arbitrary dependence will not be quite as powerful as the procedure from Theorem \ref{THM_HIER_FDR_POS_DEP}.

\begin{theorem} \label{THM_HIER_FDR_ARB_DEP} FDR Control under Arbitrary Dependence\\
The hierarchical testing procedure with critical functions
\[
\alpha_i(r) = \dfrac{\ell_i \alpha}{\ell}\frac{m_i + r-1}{m_i}\frac{1}{c_i}, \hspace{10pt} \text{where} \hspace{5pt}  c_i = 1 + \csum{j = d_i}{|\mathcal{G}_{d_i}|-1}1/ (m_i + j),
\]
for $i=1, \ldots, m$, strongly controls the FDR at level $\alpha$ under arbitrary dependence.
\end{theorem}
\vspace{10pt}

Just like Theorem \ref{THM_HIER_FDR_POS_DEP}, we consider the non-hierarchical configuration of hypotheses. In this special case, all of the critical functions are $r\alpha/(m c)$ where $c = \sum_{j=1}^{m}1/j$ so that this procedure reduces to the stepup procedure with critical function $r\alpha/(mc)$, which is the BY procedure. Thus, this result extends the BY procedure to the testing of hierarchically ordered hypotheses. On the other hand, we also consider the fixed sequence configuration. Here, the rejection threshold for $H_i$ is $m\alpha/(m-i+1)$, which is the same as the procedure from Theorem \ref{THM_HIER_FDR_POS_DEP} under this configuration.

It is easy to see that the critical functions of this procedure are scaled down compared with the procedure from Theorem \ref{THM_HIER_FDR_POS_DEP} in order to ensure the FDR control under arbitrary dependence, similar to the way the BY procedure is scaled down compared with the BH procedure. Consider the example in Figure \ref{IMG_EXAMPLE} (b) which consists of 7 hypotheses. Here, $c_1 = 1, c_2 = c_3 = 1.2,$ and $c_4 = c_5 = c_6 = c_7 = 1.76$ which means the critical functions of Theorem \ref{THM_HIER_FDR_POS_DEP} are as large, 1.2 times larger, and 1.76 times larger than the critical functions of Theorem \ref{THM_HIER_FDR_ARB_DEP} for testing $\mathcal{F}_1$, $\mathcal{F}_2$, and $\mathcal{F}_3$, respectively. The critical function of the BH procedure, on the other hand, is $\sum_{i=1}^{7}1/i = 2.59$ times larger than the critical function of the BY procedure for testing 7 hypotheses in the non-hierarchical setting. This holds in general, that the constants in the critical functions of Theorem \ref{THM_HIER_FDR_ARB_DEP} are much smaller in the hierarchical setting than in the non-hierarchical setting (i.e. the BY procedure). It shows that the FDR controlling procedure under arbitrary dependence tends to be less affected by not having the assumption of positive dependence in the hierarchical setting than in the non-hierarchical setting.

\subsection{Procedures under Block Dependence}
In this subsection, we consider block dependence and develop more powerful versions of the procedures in Theorems \ref{THM_HIER_FDR_POS_DEP} and \ref{THM_HIER_FDR_ARB_DEP} by taking this dependence into account. Since block dependence only describes the dependence of the $p$-values across families, we consider both positive dependence and arbitrary dependence to describe the dependence of the $p$-values within the families which we will refer to as block positive dependence and block arbitrary dependence, respectively.

In the fixed sequence configuration, block dependence reduces to independence. Under this configuration of hypotheses, both of our procedures presented in this subsection reduce to the more powerful FDR controlling fixed sequence procedure under independence, whereas the procedures in the last two subsections reduce to the less powerful FDR controlling fixed sequence procedure under arbitrary dependence \citep{Lynch_2014}.

First, we consider block positive dependence.

\begin{theorem} \label{THM_HIER_FDR_BLOCK_POS} FDR Control under Block Positive Dependence\\
Under Assumption \ref{ASM_POS_DEPENDENCE} and \ref{ASM_BLOCK_DEPENDENCE}, the hierarchical testing procedure with critical functions
\begin{align*}
\alpha_i(r) =
\begin{cases}
\dfrac{\ell_i r\alpha}{\ell + \ell_i(r-1)\alpha} & \text{ if $H_i$ is not a leaf hypothesis} \\
\dfrac{r\alpha}{\ell} & \text{ if $H_i$ is a leaf hypothesis}
\end{cases}
\end{align*}
for $i=1, \ldots, m$, strongly controls the FDR at level $\alpha$.
\end{theorem}

In the non-hierarchical configuration, this procedure reduces to the BH procedure since all the critical functions are $r\alpha/m$. It should be noted that under this configuration there is only one family so that block dependence is irrelevant and we are left with just the positive dependence assumption. Thus, both this procedure and the procedure from Theorem \ref{THM_HIER_FDR_POS_DEP}, which both assume positive dependence, reduce to the BH procedure in the non-hierarchical configuration.

In the hierarchical setting, this procedure offers a large improvement over the critical functions of Theorem \ref{THM_HIER_FDR_POS_DEP}. To see this, consider the binary tree configuration. In this case, the critical functions are
\begin{equation}
\alpha_i(r) = \dfrac{r\alpha}{2^{d-1}+(r-1)\alpha}, ~H_i \in \mathcal{F}_d, d = 1, \dots, D-1 \text {   and   } \alpha_i(r) = \dfrac{r\alpha}{2^{D-1}}, ~H_i \in \mathcal{F}_D. \label{EQN_BINARY_TREE_BLOCK_POS}
\end{equation}
Comparing (\ref{EQN_BINARY_TREE_POS}) to (\ref{EQN_BINARY_TREE_BLOCK_POS}), one can see that (\ref{EQN_BINARY_TREE_BLOCK_POS}) is, in general, much larger than (\ref{EQN_BINARY_TREE_POS}), approximately $r / (1 + (r-1)/(2^{D-d+1}-1))$ times larger when $\alpha$ is small. For example, when $D=5, d=3$, and $r = 4$ the increase is by a factor of almost $3$. Also, compared to Meinshausen's FWER controlling hierarchical testing procedure, which uses critical function $\ell_i\alpha/\ell, i = 1, \dots, m,$ the critical functions of Theorem \ref{THM_HIER_FDR_BLOCK_POS} are approximately $r$ times larger under every configuration for small $\alpha$. Hence, the procedure from Theorem \ref{THM_HIER_FDR_BLOCK_POS}, which requires the strongest dependence assumption to control the FDR, is our most powerful hierarchical testing procedure.

Finally, we consider block arbitrary dependence.
\begin{theorem} \label{THM_HIER_FDR_BLOCK_ARB} FDR Control under Block Arbitrary Dependence\\
Under Assumption \ref{ASM_BLOCK_DEPENDENCE}, the hierarchical testing procedure with critical functions
\begin{align*}
\alpha_i(r) =
\begin{cases}
\dfrac{\ell_i r\alpha}{\ell + \ell_i(r-1)\alpha} \dfrac{1}{c_i} & \text{ if $H_i$ is not a leaf hypothesis} \vspace{3pt}\\
\dfrac{r\alpha}{\ell} \dfrac{1}{c_i} & \text{ if $H_i$ is a leaf hypothesis}
\end{cases}
\end{align*}
where
\begin{align*}
c_i =
\begin{cases}
1 + \sum_{j = 1}^{|\mathcal{F}_{d_i}|-1}\dfrac{\ell - \ell_i\alpha}{(j + d_i)(\ell+\ell_i (j + d_i - 2) \alpha)} & \text{ if $H_i$ is not a leaf hypothesis} \vspace{3pt}\\
1 + \sum_{j = 1}^{|\mathcal{F}_{d_i}|-1}\dfrac{1}{j+d_i}& \text{ if $H_i$ is a leaf hypothesis}
\end{cases}
\end{align*}
and $|\mathcal{F}_{d_i}|$ is the cardinality of $\mathcal{F}_{d_i}$, strongly controls the FDR at level $\alpha$.
\end{theorem}

This procedure reduces to the BY procedure in the non-hierarchical configuration so that both procedures under arbitrary dependence reduce to the BY procedure. Similar to the procedures from Theorems \ref{THM_HIER_FDR_POS_DEP} and \ref{THM_HIER_FDR_ARB_DEP}, the critical functions of this procedure are a factor smaller than the critical functions of Theorem \ref{THM_HIER_FDR_BLOCK_POS}. Again, we consider the hypotheses in Figure \ref{IMG_EXAMPLE} (b). Here, $c_1 = 1, c_2 = c_3 = 1.317,$ and $c_4 = c_5 = c_6 = c_7 = 1.760$ at $\alpha = 0.05$. In this example, the $c_i$'s are significantly smaller than the constant for the BY procedure for testing 7 hypotheses, which is 2.59. However, $c_2$ and $c_3$ for Theorem \ref{THM_HIER_FDR_BLOCK_ARB} are larger than $c_2$ and $c_3$ for Theorem \ref{THM_HIER_FDR_ARB_DEP}, which are both $1.2$, but this is not true in general for the $c_i$'s. The portion of the critical function without $c_i$, is generally much larger for Theorem \ref{THM_HIER_FDR_BLOCK_ARB} than for Theorem \ref{THM_HIER_FDR_ARB_DEP} so that the procedure from Theorem \ref{THM_HIER_FDR_BLOCK_ARB} is typically more powerful than the procedure from Theorem \ref{THM_HIER_FDR_ARB_DEP}.

\begin{remark}
Our proofs of the theorems in this section heavily rely on mathematical induction. The hierarchical structure of the hypotheses implies a recursive property where the hypotheses in the subtree under any hypothesis also form a hierarchical structure. Hence, mathematical induction is a natural choice for proving results for hierarchical structures.
\end{remark}

Below, we demonstrate how the hierarchical testing procedure in Theorem \ref{THM_HIER_FDR_POS_DEP} works through an example as well as Yekutieli's and Meinshausen's hierarchical testing procedures.

\begin{table}
    \begin{center}
    \caption{The procedure from Theorem \ref{THM_HIER_FDR_POS_DEP} at level $\alpha = 0.05$ to hierarchically test the hypotheses presented in Figure \ref{IMG_EXAMPLE}(b) with $p$-values $p_1 = 0.01, p_2 = 0.75, p_3 = 0.008, p_4 = 0.6, p_5 = 0.85, p_6 = 0.03,$ and $p_7 = 0.05$.}
    \begin{tabular}{|c|c||c|c|c|c|c|c|c||l|}
        \hline
        \multicolumn{2}{|l|}{Procedure \ref{THM_HIER_FDR_BLOCK_POS}}  & $i = 1$ & $i = 2$  & $i = 3$ & $i = 4$   & $i = 5$     & $i = 6$ & $i = 7$ & Outcome\\ \noalign{\hrule height 2pt}
        \multicolumn{10}{|l|}{\textbf{Family 1}}  \\ \hline
        generalized stepup  & $ \alpha_i^*(R)$ & 0.05 & -     & -     & -     & -     & -     & - & $R = 1$    \\ \hline
        \multicolumn{10}{|l|}{\hspace{10pt} Reject $H_1$ and set $R(\mathcal{F}_{1}) = 1$}  \\ \noalign{\hrule height 3pt}
        \multicolumn{10}{|l|}{\textbf{Family 2}}  \\ \hline
        generalized stepup  & $ \alpha_i^*(R)$ & -  & 0.033 & 0.033 & -     & -     & -     & - & $R = 1$     \\ \hline
        \multicolumn{10}{|l|}{\hspace{10pt} Accept $H_2$, reject $H_3$ and set $R(\mathcal{F}_{2}) = 1$}  \\ \noalign{\hrule height 3pt}
        \multicolumn{10}{|l|}{\textbf{Family 3}}  \\ \hline
        generalized stepup  & $ \alpha_i^*(R)$ & -  & -    & -     & 0     & 0     & 0.05 & 0.05 & $R = 2$ \\ \hline
        \multicolumn{10}{|l|}{\hspace{10pt} Accept $H_4$ and $H_5$ and reject $H_6$ and $H_7$. Set $R(\mathcal{F}_{3}) = 2$}  \\ \noalign{\hrule height 3pt}
    \end{tabular}
    \label{TABLE_HIERARCH_EXMP}
    \end{center}
\end{table}
\begin{exmp}
Consider the example presented in Figure 1(b). The maximum depth of the tree is 3 and the seven hypotheses in the tree are grouped as 3 families, which are $\{H_1\}$, $\{H_2, H_3\}$, and $\{H_4, H_5, H_6, H_7\}$. Suppose the $p$-values are $p_1 = 0.01, p_2 = 0.75, p_3 = 0.008, p_4 = 0.6, p_5 = 0.85, p_6 = 0.03$, and $p_7 = 0.05$ and the hypotheses are tested using the procedure from Theorem \ref{THM_HIER_FDR_POS_DEP}, Yekutieli's procedure, and Meinshausen's procedure at level $\alpha = 0.05$.

Table \ref{TABLE_HIERARCH_EXMP} shows the value of the variables step-by-step for the procedure from Theorem \ref{THM_HIER_FDR_POS_DEP}. The first family is tested using the generalized stepup procedure and $H_1$, the only hypothesis in this family, is rejected. Now, $R(\mathcal{F}_1) = 1$. The second family is tested using the generalized stepup procedure with critical functions $\alpha_2^*(r) = \alpha_2(r+1)$ and $\alpha_3^*(r) = \alpha_3(r+1)$. $H_3$ can be rejected but $H_2$ cannot. Thus, $R(\mathcal{F}_2) = 1$. Finally, the third family is tested. Since $H_2$ was accepted and $H_3$ was rejected, we have $\alpha^*_4(r) = \alpha^*_5(r) = 0$, $\alpha^*_6(r) = \alpha_6(r+2)$, and $\alpha^*_7(r) = \alpha_7(r+2)$. Hypotheses $H_6$ and $H_7$ are rejected by the generalized stepup procedure.

Yekutieli's hierarchical testing procedure groups the hypotheses into families that share the same parent hypothesis so that the 4 families are $\{H_1\}$, $\{H_2, H_3\}$, $\{H_4, H_5\}$, and $\{H_6, H_7\}$. This procedure rejects hypotheses $H_1$ and $H_3$ (Table \ref{TABLE_HIERARCH_EXMP_YEKUT}). Meinshausen's hierarchical testing procedure uses a fixed rejection threshold $\ell_i\alpha/\ell$ for testing $H_i$. This procedure rejects $H_1$ and $H_3$ (Table \ref{TABLE_HIERARCH_EXMP_MEIN}).
\begin{table}
	\begin{center}
    \caption{Yekutieli's procedure at level $\alpha = 0.05$ to hierarchically test the hypotheses presented in Figure \ref{IMG_EXAMPLE}(b) with $p$-values $p_1 = 0.01, p_2 = 0.75, p_3 = 0.008, p_4 = 0.6, p_5 = 0.85, p_6 = 0.03,$ and $p_7 = 0.05$.} \vskip 10pt
    \begin{tabular}{|c|c||c|c|c|c|c|c|c||l|}
        \hline
        \multicolumn{2}{|l|}{Yekutieli's Procedure}   & $i = 1$ & $i = 2$  & $i = 3$ & $i = 4$   & $i = 5$     & $i = 6$ & $i = 7$ & Outcome\\ \noalign{\hrule height 2pt}
        \multicolumn{10}{|l|}{\textbf{Family 1}}  \\ \hline
        BH procedure  & $ \alpha_i(R)$ & 0.0174 & -     & -     & -     & -     & -     & - & $R = 1$    \\ \hline
        \multicolumn{10}{|l|}{\hspace{10pt} Reject $H_1$}  \\ \noalign{\hrule height 3pt}
        \multicolumn{10}{|l|}{\textbf{Family 2}}  \\ \hline
        BH procedure  & $ \alpha_i(R)$ & -  & 0.009       & 0.009     & -     & -     & -     & - & $R = 1$     \\ \hline
        \multicolumn{10}{|l|}{\hspace{10pt} Accept $H_2$ and reject $H_3$}  \\ \noalign{\hrule height 3pt}
        \multicolumn{10}{|l|}{\textbf{Family 3}}  \\ \hline
        Not Tested    &              & -  & -       & -     & -     & -     & -     & - &      \\ \hline
        \multicolumn{10}{|l|}{\hspace{10pt} Accept $H_4$ and $H_5$}  \\ \noalign{\hrule height 3pt}
        \multicolumn{10}{|l|}{\textbf{Family 4}}  \\ \hline
        BH procedure  & $ \alpha_i(R)$ & -  & -       & -     & -     & -     & 0      & 0 & $R = 0$  \\ \hline
        \multicolumn{10}{|l|}{\hspace{10pt} Accept $H_6$ and $H_7$}  \\ \noalign{\hrule height 3pt}
    \end{tabular}
    \label{TABLE_HIERARCH_EXMP_YEKUT}
    \end{center}
\end{table}
\begin{table}
        \begin{center}
    \caption{Meinshausen's procedure at level $\alpha = 0.05$ to hierarchically test the hypotheses presented in Figure \ref{IMG_EXAMPLE}(b) with $p$-values $p_1 = 0.01, p_2 = 0.75, p_3 = 0.008, p_4 = 0.6, p_5 = 0.85, p_6 = 0.03,$ and $p_7 = 0.05$.}
    \begin{tabular}{|c|c||c|c|c|c|c|c|c||l|}
        \hline
        \multicolumn{2}{|l|}{Meinshausen's Procedure}   & $i = 1$ & $i = 2$  & $i = 3$ & $i = 4$   & $i = 5$     & $i = 6$ & $i = 7$ & Outcome\\ \noalign{\hrule height 2pt}
        \multicolumn{10}{|l|}{\textbf{Family 1}}  \\ \hline
        Single Step  & $ \alpha_i$ & 0.05 & -     & -     & -     & -     & -     & - & $R = 1$    \\ \hline
        \multicolumn{10}{|l|}{\hspace{10pt} Reject $H_1$}  \\ \noalign{\hrule height 3pt}
        \multicolumn{10}{|l|}{\textbf{Family 2}}  \\ \hline
        Single Step  & $ \alpha_i$ & -  & 0.025   & 0.025  & -     & -     & -     & - & $R = 1$     \\ \hline
        \multicolumn{10}{|l|}{\hspace{10pt} Accept $H_2$ and reject $H_3$ }  \\ \noalign{\hrule height 3pt}
        \multicolumn{10}{|l|}{\textbf{Family 3}}  \\ \hline
        Single Step  & $ \alpha_i$ & -  & -       & -    & -     & -     & 0.0125     & 0.0125 & $R = 0$     \\ \hline
        \multicolumn{10}{|l|}{\hspace{10pt} Accept $H_4,H_5,H_6$, and $H_7$}  \\ \noalign{\hrule height 3pt}
    \end{tabular}
    \label{TABLE_HIERARCH_EXMP_MEIN}
    \end{center}
\end{table}
\end{exmp}

\section{Simulation Study}

We conducted a simulation study to evaluate the performance of the proposed procedures. Specifically, the simulation study compared the performance of our proposed procedures, which are labeled Procedures \ref{THM_HIER_FDR_POS_DEP}-\ref{THM_HIER_FDR_BLOCK_ARB} corresponding to the procedures introduced in Theorems \ref{THM_HIER_FDR_POS_DEP}-\ref{THM_HIER_FDR_BLOCK_ARB}, against Yekutieli's FDR controlling procedure in terms of the FDR control and average power. Several dependence configurations were considered as well as different hierarchical structures.

We generated $m$ normal random variables with covariance matrix $\Sigma$ and mean vector $\vec{\mu} = (\mu_1, \dots, \mu_{m})$ to test the $m$ hypotheses $H_i: \mu_i \le 0$ versus $H_i': \mu_i > 0, i = 1, \dots, m$. When $H_i$ was true, we set $\mu_i = 0$. When $H_i$ was false, we set $\mu_i$ to a positive value which was non-increasing in its depth $d_i$. Our intention was to simulate the setting where hypotheses that are near the top of the hierarchy are easier to reject than hypotheses near the bottom. As for the joint dependence, we considered a common correlation structure where $\Sigma$ had off-diagonal components equal to $\rho$ and diagonal components equal to 1. The $p$-value for testing the $i^{th}$ leaf hypothesis was calculated using a one sided, one-sample Z-test.

We constructed two types of hierarchies: a shallow hierarchy and a deep hierarchy. Both hierarchies had 1000 leaf hypotheses.

The leaf hypotheses were randomly chosen with probability $\pi_0$ to be true and $1 - \pi_0$ to be false. Each non-leaf hypothesis was set to true only if all of its child hypotheses were true; otherwise it was set to false. For both hierarchies, the tree was balanced so that each parent hypothesis had the same number of child hypotheses. The two hierarchies are described in detail below.

\vskip 10pt
\noindent
\textbf{Shallow Hierarchy}: The maximum depth of this tree is 2 so that a hypothesis is either a leaf hypothesis or a top-level hypothesis with no parent. There are 10 top-level hypotheses each of which have 100 child hypotheses giving a total of 1010 hypotheses. For each false hypothesis $H_i$, $\mu_i = 3$ if $d_i = 1$ and $\mu_i = 2$ if $d_i = 2$.

\vskip 5pt
\noindent
\textbf{Deep Hierarchy}: The maximum depth of this tree is 4 and there are 8 top-level parents. Each parent hypothesis has 5 child hypotheses giving a total of 1248 hypotheses. For each false hypothesis $H_i$, $\mu_i = 3.5$ if $d_i = 1$, $\mu_i = 3$ if $d_i = 2$ or $3$, and $\mu_i = 2$ if $d_i = 4$.

\vskip 10pt
We set $\alpha = 0.05$ and for each procedure, we noted the false discovery proportion, which is the proportion of falsely rejected hypotheses among all rejected hypotheses, and the the proportion of rejected false null hypotheses among all false null hypotheses. Each tree was generated and tested 5000 times and the simulated values of the FDR and average power were obtained by averaging out the 5000 values of these two proportions, respectively.

\begin{figure}[t]
    \begin{center}
        \includegraphics[scale=.65,angle=270, trim=0 50 0 50]{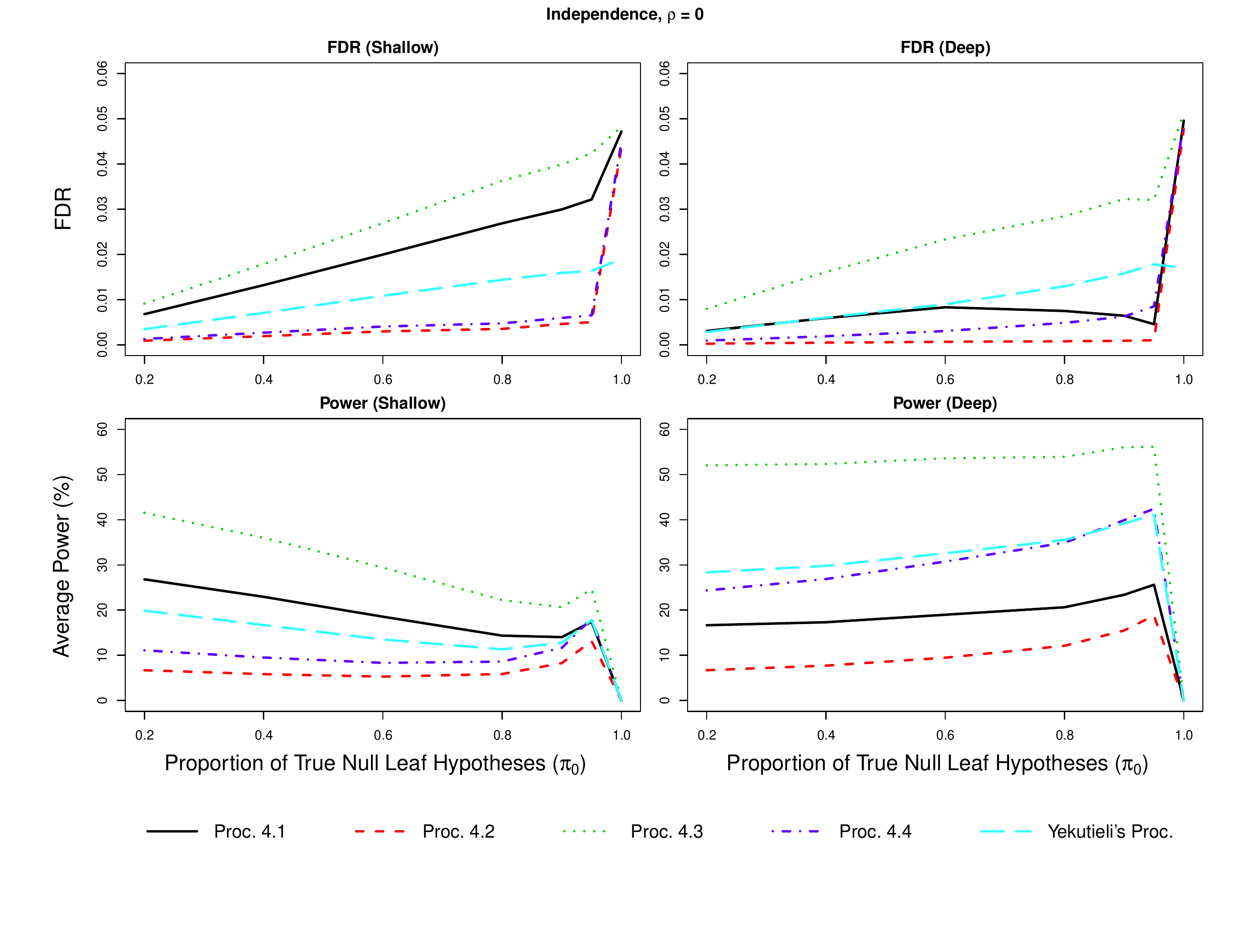}
        \vskip -45pt
        \caption{\textit{FDR (top row) and average power (bottom row) of Procedures \ref{THM_HIER_FDR_POS_DEP} (solid line), \ref{THM_HIER_FDR_ARB_DEP} (dashed), \ref{THM_HIER_FDR_BLOCK_POS} (dotted), \ref{THM_HIER_FDR_BLOCK_ARB} (dot dash), and Yekutieli's procedure (long dash) under independence for the shallow hierarchy (left column) and the deep hierarchy (right column) where the proportion of true null leaf hypotheses varies from 0.2 to 1.}}
        \label{IMG_INDEPENDENCE}
    \end{center}
\end{figure}

\begin{figure}[h!]
    \begin{center}
        \includegraphics[scale=.65,angle=270, trim=0 50 0 50]{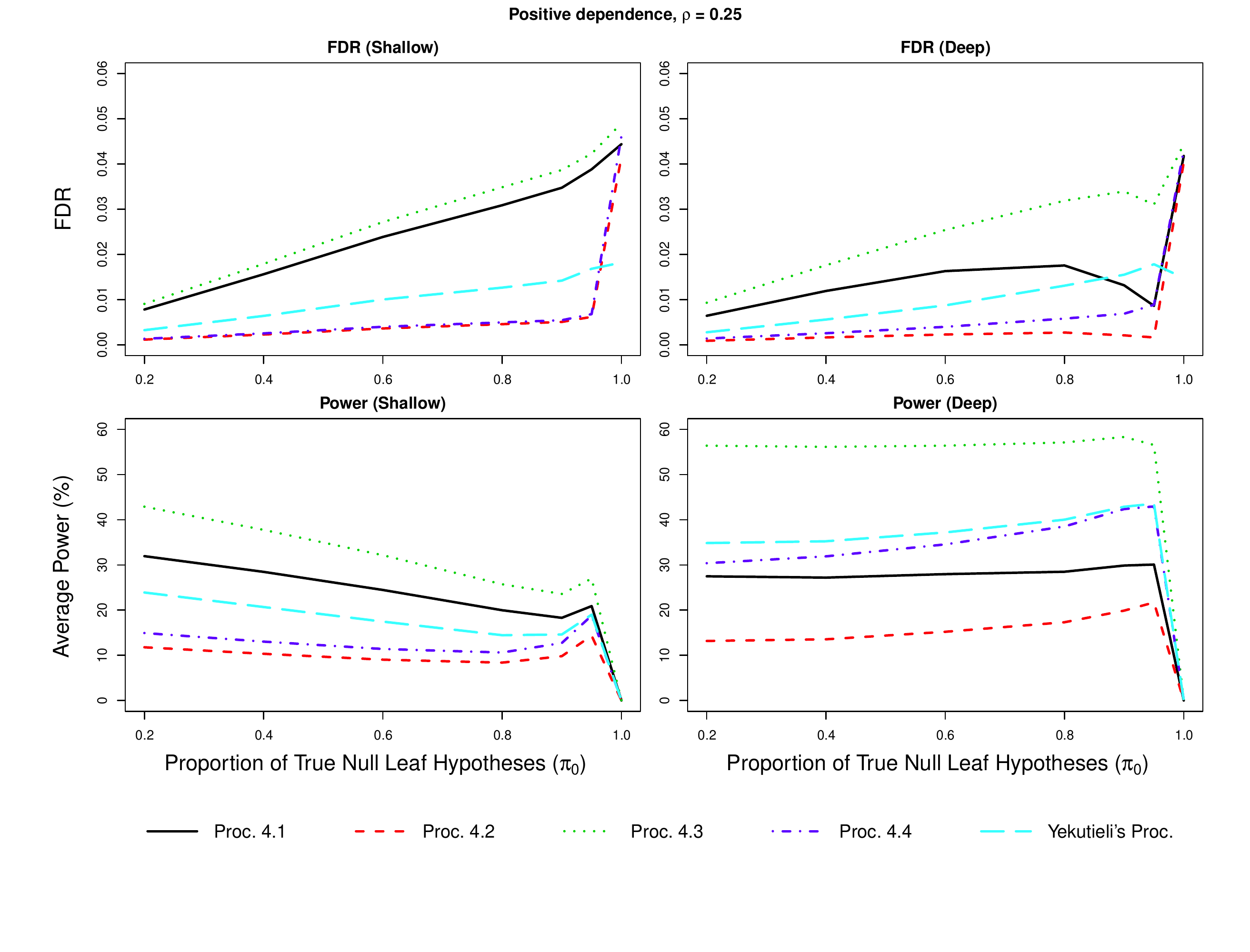}
        \vskip -45pt
        \caption{\textit{FDR (top row) and average power (bottom row) of Procedures \ref{THM_HIER_FDR_POS_DEP} (solid line), \ref{THM_HIER_FDR_ARB_DEP} (dashed), \ref{THM_HIER_FDR_BLOCK_POS} (dotted), \ref{THM_HIER_FDR_BLOCK_ARB} (dot dash), and Yekutieli's procedure (long dash) under common correlation with $\rho = 0.25$ for the shallow hierarchy (left column) and the deep hierarchy (right column) where the proportion of true null leaf hypotheses varies from 0.2 to 1.}}
        \label{IMG_POS_DEPENDENCE_25}
    \end{center}
\end{figure}

\begin{figure}[h!]
    \begin{center}
        \includegraphics[scale=.65,angle=270, trim=0 50 0 50]{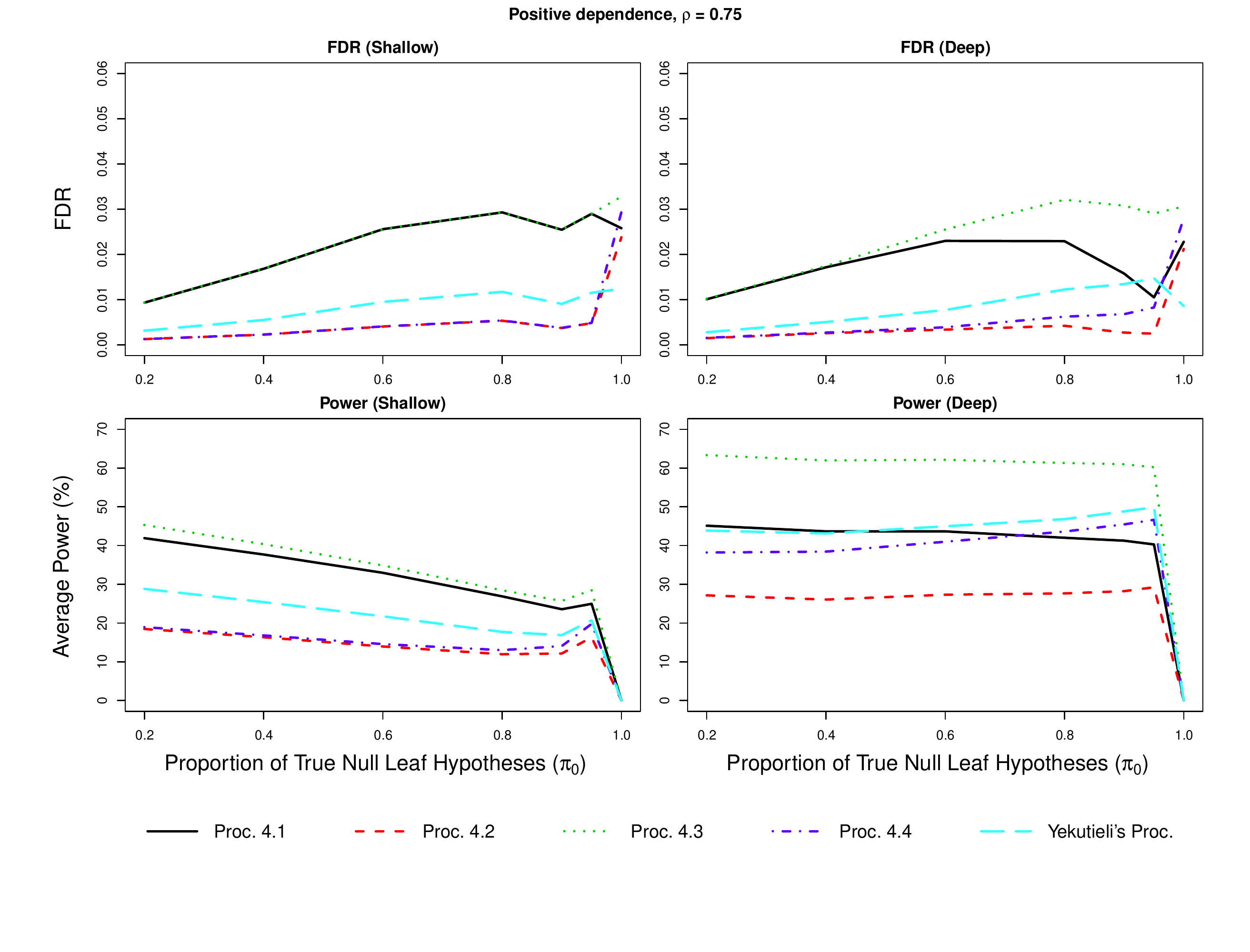}
        \vskip -45pt
        \caption{\textit{FDR (top row) and average power (bottom row) of Procedures \ref{THM_HIER_FDR_POS_DEP} (solid line), \ref{THM_HIER_FDR_ARB_DEP} (dashed), \ref{THM_HIER_FDR_BLOCK_POS} (dotted), \ref{THM_HIER_FDR_BLOCK_ARB} (dot dash), and Yekutieli's procedure (long dash) under common correlation with $\rho = 0.75$ for the shallow hierarchy (left column) and the deep hierarchy (right column) where the proportion of true null leaf hypotheses varies from 0.2 to 1.}}
        \label{IMG_POS_DEPENDENCE_75}
    \end{center}
\end{figure}

Figure \ref{IMG_INDEPENDENCE} displays the FDR and average power under independence as $\pi_0$ varies from 0.2 to 1. As seen from Figure \ref{IMG_INDEPENDENCE}, all the procedures control the FDR at level 0.05. In terms of power, Procedure \ref{THM_HIER_FDR_BLOCK_POS} outperforms Yekutieli's procedure quite substantially and in some cases even doubles the power of Yekutieli's procedure. Procedure \ref{THM_HIER_FDR_POS_DEP}, which controls the FDR under positive dependence, outperforms Yekutieli's procedure under the shallow hierarchy but is outperformed by Yekutieli's procedure in the deep hierarchy. In the deep hierarchy, Procedure \ref{THM_HIER_FDR_BLOCK_ARB} and Yekutieli's procedure are comparable in terms of power. Not surprisingly, Procedure \ref{THM_HIER_FDR_ARB_DEP}, which controls the FDR under arbitrary dependence, performs the worst.

Figures \ref{IMG_POS_DEPENDENCE_25} and \ref{IMG_POS_DEPENDENCE_75} display the FDR and average power under common correlation with $\rho = 0.25$ and $\rho = 0.75$, respectively, as $\pi_0$ varies from 0.2 to 1. The FDRs of all the procedures are controlled at level $0.05$ under both weak and strong correlation. It should be noted that assumption \ref{ASM_BLOCK_DEPENDENCE} (block dependence) does not hold under this dependence configuration but Procedures \ref{THM_HIER_FDR_BLOCK_POS} and \ref{THM_HIER_FDR_BLOCK_ARB} still control the FDR suggesting that both procedures are fairly robust to departures from this assumption. In terms of power, these figures show a similar pattern to Figure \ref{IMG_INDEPENDENCE} where Procedure \ref{THM_HIER_FDR_BLOCK_POS} is the most powerful and Procedure \ref{THM_HIER_FDR_ARB_DEP} is the least powerful. The remaining three procedures fall somewhere in the middle depending on the setting.

\section{Real Data Analysis}

We applied our proposed procedures as well as Yekutieli's procedure to a real data set. We used the data set of \cite{Caporaso_2011}, available in the phyoseq Bioconductor package at www.bioconductor.org, which provides the abundances of individual microbes in different ecological environments as well as their phylogenetic relationships. The data can be naturally organized into a hierarchy consisting of taxonomic units according to their phylogenetic relationships. The question of interest is whether there is an association between a taxonomic unit and ecological environment. Specifically, we tested the null hypothesis that the mean abundance for the taxonomic unit is the same across environments versus the alternative hypotheses that the mean abundance for the taxonomic unit is different across environments. The $p$-value for each hypothesis was determined by using an F-test where the abundance for a taxonomic unit in a given environment was determined based on the total abundance of each microbe within the taxonomic unit for the given environment (for more information see \cite{Sankaran_2013}).

We restricted our analysis to the microbes in the Actinobacteria phylum which had 1631 individual microbes. The taxonomic hierarchy in the Actinobacteria phylum consisted of 3261 taxonomic units so that the total number of hypotheses is 3261 across 39 families. We tested the hypotheses at various significance levels and the number of rejections for each procedure are displayed in Table \ref{TABLE_REJECTIONS}. All of the procedures are seen to make a substantial number of discoveries, even when $\alpha = 0.01$. In terms of the number of rejections, one can easily see that Procedure \ref{THM_HIER_FDR_BLOCK_POS} is by far the best, significantly outperforming the other procedures. Procedure \ref{THM_HIER_FDR_BLOCK_ARB} outperforms Yekutieli's procedure when $\alpha$ is moderate to large but Yekutieli's procedure outperforms Procedure \ref{THM_HIER_FDR_BLOCK_ARB} when $\alpha$ is small. Procedure \ref{THM_HIER_FDR_ARB_DEP} is, not surprisingly, the worst since it is the only procedure that controls the FDR under arbitrary dependence.
\begin{table}[h!]
    \caption{The number of rejections out of 3261 hypotheses by Procedures \ref{THM_HIER_FDR_POS_DEP}, \ref{THM_HIER_FDR_ARB_DEP}, \ref{THM_HIER_FDR_BLOCK_POS}, \ref{THM_HIER_FDR_BLOCK_ARB}, and Yekutieli's procedure at various significance levels for the microbe abundance data set of \cite{Caporaso_2011} restricted to the Actinobacteria phylum.} \vskip -15pt
    \begin{center}
    \begin{tabular}{|c||c|c|c|c|c|}
        \hline
        \textbf{$\alpha$} & \textbf{Procedure \ref{THM_HIER_FDR_POS_DEP}} & \textbf{Procedure \ref{THM_HIER_FDR_ARB_DEP}} & \textbf{Procedure \ref{THM_HIER_FDR_BLOCK_POS}} & \textbf{Procedure \ref{THM_HIER_FDR_BLOCK_ARB}} & \textbf{Yekutieli's Procedure}\\ \noalign{\hrule height 2pt}
          0.01  & 75  & 68  & 144  & 107 & 123 \\ \hline
          0.025 & 88  & 75  & 574  & 148 & 165 \\ \hline
          0.05  & 118 & 92  & 1156 & 353 & 230 \\ \hline
          0.1   & 138 & 108 & 1497 & 813 & 253 \\ \hline
    \end{tabular}
    \label{TABLE_REJECTIONS}
    \end{center}
\end{table}

\section{Conclusion}

In this paper, we have developed several FDR controlling procedures for testing hierarchically ordered hypotheses. To our knowledge, we have, for the first time, presented hierarchical testing methods with proven FDR control under dependence. Furthermore, we have developed a method which controls the FDR under block positive dependence and in our simulation study, it was shown to be more powerful than Yekutieli's hierarchical testing  procedure and other proposed procedures. A particularly interesting aspect of this work is that we have connected two contrasting testing methods in the proposed hierarchical testing methods: fixed sequence procedures, which assume the hypotheses have a fixed pre-defined testing order, and stepwise procedures, which do not assume the hypotheses having any pre-defined testing order.

We believe in this paper we have made a significant step in terms of multiple testing with structured hypotheses. The techniques developed in this paper can be used to develop procedures to test hypotheses with more complex hierarchical structures where hypotheses are not restricted to only one parent. Such procedures would have applications towards testing interaction hypotheses, for example in gene expression data, where main effects are tested first and pairwise interactions are tested only if the two main effects making up the interaction are significant.

\section{Appendix}
Let us first state and prove the following lemmas which are used in the proofs of Theorems \ref{THM_HIER_FDR_POS_DEP}, \ref{THM_HIER_FDR_ARB_DEP}, \ref{THM_HIER_FDR_BLOCK_POS}, and \ref{THM_HIER_FDR_BLOCK_ARB}.

\begin{lemma} \label{LEMMA_POSITIVE_DEPENDENCE}
Under Assumption \ref{ASM_POS_DEPENDENCE}, if $\Gamma(P_1, \dots, P_m)$ is a discrete coordinatewise non-increasing function of the $p$-values taking on values $\gamma_1 < \dots < \gamma_n$ and $t(\cdot)$ is a non-decreasing function on $\{\gamma_1, \ldots, \gamma_n\}$, then for each true null $H_j$,
\[
\sum_{i=1}^{n}\pr{\Gamma = \gamma_i \cond P_j \le t(\gamma_i)} \le \pr{\Gamma \ge \gamma_1 \cond P_j \le t(\gamma_1)}.
\]
\end{lemma}
\begin{proof}[Proof of Lemma \ref{LEMMA_POSITIVE_DEPENDENCE}]
\begin{align*}
    &\sum_{i=1}^{n}\pr{\Gamma = \gamma_i \cond P_j \le t(\gamma_i)} \\
=   &\sum_{i=1}^{n}\pr{\Gamma \ge \gamma_i \cond P_j \le t(\gamma_i)} - \sum_{i=1}^{n-1}\pr{\Gamma \ge \gamma_{i+1} \cond P_j \le t(\gamma_i)} \\
=   & ~\pr{\Gamma \ge \gamma_1 \cond P_j \le t(\gamma_1)} - \sum_{i=2}^{n}\left[\pr{\Gamma \ge \gamma_i \cond P_j \le t(\gamma_{i-1})} - \pr{\Gamma \ge \gamma_i \cond P_j \le t(\gamma_i)}\right] \\
\le & ~\pr{\Gamma \ge \gamma_i \cond P_j \le t(\gamma_1)}.
\end{align*}
The inequality follows by Assumption \ref{ASM_POS_DEPENDENCE}.
\end{proof}

\begin{lemma} \label{LEMMA_ARBITRARY_DEPENDENCE}
Under arbitrary dependence of the $p$-values, if $\Gamma(P_1, \dots, P_m)$ is a discrete function of the $p$-values taking on values $\gamma_1 < \dots < \gamma_n$ and $t(\cdot)$ is a positive non-decreasing function on $\{\gamma_1, \ldots, \gamma_n\}$ with the convention that $t(\gamma_0) = 0$, then for each true null $H_j$,
\[
\sum_{i=1}^{n}\frac{1}{t(\gamma_i)}\pr{\Gamma = \gamma_i,  P_j \le t(\gamma_i)} \le \sum_{i=1}^{n}\frac{t(\gamma_i) - t(\gamma_{i-1})}{t(\gamma_i)}.
\]
\end{lemma}
\begin{proof}[Proof of Lemma \ref{LEMMA_ARBITRARY_DEPENDENCE}]
Using the convention that $0/0 = 0$, we have
\begin{align*}
    &~\sum_{i=1}^{n}\frac{1}{t(\gamma_i)}\pr{\Gamma = \gamma_i, P_j \le t(\gamma_i)} \\
=   &~\sum_{i=1}^{n}\left[\frac{1}{t(\gamma_i)}\pr{\Gamma \ge \gamma_i, P_j \le t(\gamma_i)} - \frac{1}{t(\gamma_{i-1})}\pr{\Gamma \ge \gamma_i, P_j \le t(\gamma_{i-1})}\right] \\
\le &~\sum_{i=1}^{n}\frac{1}{t(\gamma_i)}\pr{\Gamma \ge \gamma_i, t(\gamma_{i-1}) < P_j \le t(\gamma_i)} \\
\le &~\sum_{i=1}^{n}\frac{1}{t(\gamma_i)}\pr{t(\gamma_{i-1}) < P_j \le t(\gamma_i)} \\
=   &~\sum_{i=1}^{n-1}\left(\frac{1}{t(\gamma_i)} - \frac{1}{t(\gamma_{i+1})}\right) \pr{P_j \le t(\gamma_i)} + \frac{1}{t(\gamma_n)}\pr{P_j \le t(\gamma_n)} \\
\le &~\sum_{i=1}^{n}\frac{t(\gamma_i) - t(\gamma_{i-1})}{t(\gamma_i)}. ~\qedhere
\end{align*}
\end{proof}

\subsection{Proof of Proposition \ref{PROP_ALGORITHM}}
Assume $k > \psi(k)$. Then, step 2(a) of Algorithm 3.1 is repeated until for some $\ell \ge 2$, $r_\ell \le \psi(r_\ell)$. For $t = 1, \dots, \ell-1$, we have $r_t > \psi(r_t)$ implying $r_t > r_{t+1}$. Thus, $r_\ell < r_1 = k$. For any integer $r$ from 0 to $k-1$ such that $r \le \psi(r)$, we will show that $r_\ell \ge r$. To prove it, we show using induction that $r_t \ge r, t = 1, \dots, \ell$. Since $r_1 = k > r$, by induction assume $r_{t-1} \ge r$. Then, $r_t = \psi(r_{t-1}) \ge \psi(r) \ge r$. Since $r_\ell \le \psi(r_\ell)$, $r_\ell < k$, and $r_\ell \ge r$, we have $R = r_\ell = \max\{0 \le r \le k-1: r \le \psi(r)\}$.

Conversely, assume $k \le \psi(k)$. Then, step 2(b) is repeated until for some $\ell \ge 2$, $r_\ell > \psi(r_\ell)$. For $t = 1, \dots, \ell-1$, we have $r_t \le \psi(r_t)$ implying $r_t < \psi(r_t)+1 = r_{t+1}$. Thus, $r_\ell > r_1 = k$. For any integer $r$ from $k+1$ to $m+1$ such that $r > \psi(r)$, we will show that $r_\ell \le r$. To prove it, we show using induction that $r_t \le r, t = 1, \dots, \ell$. Since $r_1 = k < r$, by induction assume $r_{t-1} \le r$. Then, $r_t = \psi(r_{t-1})+1 \le \psi(r)+1 \le r$. Since $r_\ell > \psi(r_\ell)$, $r_\ell > k$, and $r_\ell \le r$, we have $R = r_\ell-1 = \min\{k+1 \le r \le m+1 : r > \psi(r)\}-1$. \qed

\subsection{Proof of Theorem \ref{THM_HIER_FDR_POS_DEP}}
In this proof and the remaining proofs, we will use the convention that $0/0 = 0$. For convenience of notation, define $\mathcal{G}_d = \bigcup_{j=1}^{d}\mathcal{F}_j$\label{NOTATION_G} and $R(\mathcal{G}_d)$ is the number of rejections in the first $d$ families, $\mathcal{F}_j, j = 1, \dots, d$. Let $|\mathcal{G}_d|$ be the cardinality of $\mathcal{G}_d$.

We will show that
\begin{equation}
\expv{\frac{\V{\mathcal{M}_i}}{R}} \le \frac{\ell_i \alpha}{\ell}, i = 1, \dots, m. \label{EQN_POS_DEP_1}
\end{equation}
\emph{Proof of (\ref{EQN_POS_DEP_1}).}
The event $\{H_i \text{ is rejected}\}$ implies all ancestors of $H_i$ are rejected so there must be at least $d_i$ rejections in the first $d_i$ families. Therefore, the event $\{H_i \text{ is rejected}\}$ implies the following two inequalities:
\begin{align}
&d_i \le \Rd \le \cardgdi, \label{EQN_POS_DEP_4} \\
&\Rd-1 \le R - \R{\mathcal{M}_i}. \label{EQN_POS_DEP_5}
\end{align}
The second inequality follows from the fact that $\mathcal{G}_{d_i} / \{H_i\} \subseteq \mathcal{M}/\mathcal{M}_i$ so that $R(\mathcal{G}_{d_i} / \{H_i\}) \le R(\mathcal{M}/\mathcal{M}_i)$.

If $H_i$ is true,
\begin{eqnarray}
&    &\expv{\frac{\V{\mathcal{M}_i}}{R}} \nonumber
 \le  \expv{\frac{\V{\mathcal{M}_i}}{\V{\mathcal{M}_i} + R - \R{\mathcal{M}_i}}} \nonumber \\
&\le &\expv{\frac{m_i}{m_i + R - \R{\mathcal{M}_i}}\ind{H_i \text{ is rejected}}} \nonumber \\
& \le &   \expv{\frac{m_i}{m_i + \Rd-1}\ind{H_i \text{ is rejected}}} \nonumber \\
&=   &\csum{r = d_i}{\cardgdi}\expv{\frac{m_i}{m_i + r -1}\ind{\Rd = r, H_i \text{ is rejected}}} \nonumber \\
& \le & \csum{r = d_i}{\cardgdi}\frac{m_i}{m_i + r-1}\pr{\Rd = r, P_i \le \alpha_i(r)} \nonumber \\
&\le &\csum{r = d_i}{\cardgdi}\frac{m_i \alpha_i(r)}{m_i + r-1}\pr{\Rd = r \cond P_i \le \alpha_i(r)} \nonumber \\
& = &   \frac{\ell_i\alpha}{\ell}\csum{r = d_i}{\cardgdi}\pr{\Rd = r \cond P_i \le \alpha_i(r)}. \label{EQN_POS_DEP_2}
\end{eqnarray}
The second inequality follows from the fact that $\V{\mathcal{M}_i} \le m_i$ and $\V{\mathcal{M}_i} / (\V{\mathcal{M}_i} + R - \R{\mathcal{M}_i})$ is an increasing function of $\V{\mathcal{M}_i}$. The third inequality follows from (\ref{EQN_POS_DEP_5}) and the first equality follows from (\ref{EQN_POS_DEP_4}). The fourth inequality follows by the fact that the event $\{H_i \text{ is rejected}\} = \{H_{T(i)} \text{ is rejected}, P_i \le \alpha_i(\Rd)\}$.

Since the number of rejections by the generalized stepup procedure is a coordinatewise non-increasing function of the $p$-values, it follows that $\Rd$ is also a coordinatewise non-increasing function of the $p$-values. Therefore, by Lemma \ref{LEMMA_POSITIVE_DEPENDENCE},
\begin{eqnarray}
&    &\csum{r = d_i}{\cardgdi}\pr{\Rd = r \cond P_i \le \alpha_i(r)} \le \pr{\Rd \ge d_i \cond P_i \le \alpha_i(d_i)} \le 1. \label{EQN_POS_DEP_3}
\end{eqnarray}
From (\ref{EQN_POS_DEP_3}), we have that (\ref{EQN_POS_DEP_2}) is less than $\ell_i\alpha/\ell$. Thus, (\ref{EQN_POS_DEP_1}) holds when $H_i$ is true.

We will use induction to show (\ref{EQN_POS_DEP_1}) also holds when $H_i$ is false. When $H_i$ is a false null leaf hypothesis, then (\ref{EQN_POS_DEP_1}) is true trivially. Otherwise, assume (\ref{EQN_POS_DEP_1}) is true for every false child hypothesis of $H_i$. Thus, (\ref{EQN_POS_DEP_1}) is true for  all children of $H_i$. We note that when $H_i$ is false, $\V{\mathcal{M}_i} = \sum_{j:T(j) = i}\V{\mathcal{M}_j}$ and
\[
\expv{\frac{\V{\mathcal{M}_i}}{R}} = \csum{j:T(j) = i}{}\expv{\frac{\V{\mathcal{M}_j}}{R}} \le \csum{j:T(j) = i}{}\frac{\ell_j\alpha}{\ell} = \frac{\ell_i \alpha}{\ell}.
\]
Thus, (\ref{EQN_POS_DEP_1}) holds for all true and false null hypotheses.

\vspace{5pt}
\noindent\emph{Proof of Theorem \ref{THM_HIER_FDR_POS_DEP}.} By (\ref{EQN_POS_DEP_1}), we have
\[
\text{FDR} = \csum{i:T(i) = 0}{}\expv{\frac{\V{\mathcal{M}_i}}{R}} \le \csum{i:T(i) = 0}{}\frac{\ell_i \alpha}{\ell} = \alpha. ~\qed
\]

\subsection{Proof of Theorem \ref{THM_HIER_FDR_ARB_DEP}}
We will show that (\ref{EQN_POS_DEP_1}) holds under arbitrary dependence for the procedure introduced in Theorem \ref{THM_HIER_FDR_ARB_DEP}.

\vspace{10pt}
\noindent\emph{Proof of (\ref{EQN_POS_DEP_1}).} When $H_i$ is true, by the fourth inequality of (\ref{EQN_POS_DEP_2}), we have
\begin{eqnarray*}
&& \expv{\frac{\V{\mathcal{M}_i}}{R}}
\le  \csum{r = d_i}{\cardgdi}\frac{m_i}{m_i + r-1}\pr{\Rd = r, P_i \le \alpha_i(r)} \\
&=& \frac{\ell_i\alpha}{\ell}\frac{1}{c_i}\csum{r = d_i}{\cardgdi}\frac{1}{\alpha_i(r)}\pr{\Rd = r, P_i \le \alpha_i(r)} \\
& \le &  \frac{\ell_i\alpha}{\ell}\frac{1}{c_i}\left(1 + \csum{r = d_i+1}{\cardgdi}\frac{\alpha_i(r) - \alpha_i(r-1)}{\alpha_i(r)}\right) \\
&=& \frac{\ell_i\alpha}{\ell}.
\end{eqnarray*}
The second inequality follows by Lemma \ref{LEMMA_ARBITRARY_DEPENDENCE}. Thus, (\ref{EQN_POS_DEP_1}) holds when $H_i$ is true. When $H_i$ is false, (\ref{EQN_POS_DEP_1}) also holds by the same argument used in the proof of Theorem \ref{THM_HIER_FDR_POS_DEP}. Hence, (\ref{EQN_POS_DEP_1}) holds for all hypotheses.

\vspace{10pt}
\noindent\emph{Proof of Theorem \ref{THM_HIER_FDR_ARB_DEP}.} Since (\ref{EQN_POS_DEP_1}) holds for each $i = 1, \dots, m$, FDR control follows by the same argument used in the proof of Theorem \ref{THM_HIER_FDR_POS_DEP}. \qed

\subsection{Proof of Theorem \ref{THM_HIER_FDR_BLOCK_POS}} Recursively define the random variables $A_1, \dots, A_m$ as follows:
\begin{align*}
A_i =
\begin{cases}
1 &\text{ $T(i) = 0$}, \\
A_{T(i)} &\text{$T(i) \neq 0$ and $H_{T(i)}$ is false}, \\
A_{T(i)}(1 - (1/R({\mathcal{G}_{d_{T(i)}}}))\ind{H_{T(i)} \text{ is rejected}} &\text{$T(i) \neq 0$ and $H_{T(i)}$ is true.}
\end{cases}
\end{align*}
Notice that $A_i$ is a function of the $p$-values corresponding to the hypotheses in families $\mathcal{F}_1, \dots, \mathcal{F}_{d_i-1}$ so that $P_i$ and $A_i$ are independent due to Assumption \ref{ASM_BLOCK_DEPENDENCE}.

When $H_i$ is a true null hypothesis, we have the following useful inequality
\begin{equation}
\expv{A_i \frac{\ind{H_i \text{ is rejected}}}{\alpha_i(\Rd)}} \le \expv{A_i}. \label{EQN_FDR_INDEP_3}
\end{equation}
\emph{Proof of (\ref{EQN_FDR_INDEP_3}).} With the convention $R(\mathcal{G}_0) = 0$, we have
\begin{eqnarray*}
   &  &\expv{A_i \frac{\ind{H_i \text{ is rejected}}}{\alpha_i(\Rd)}}
 =     \expv{A_i \frac{\ind{H_{T(i)} \text{ is rejected}, P_i \le \alpha_i(R(\mathcal{G}_{d_{i}}))}}{\alpha_i(R(\mathcal{G}_{d_{i}}))}} \nonumber \\
&=    &\expv{A_i \frac{\ind{H_{T(i)} \text{ is rejected}, P_i \le \alpha_i(R(\mathcal{G}_{d_i-1} + R(\mathcal{F}_{d_i}))}}{\alpha_i(R(\mathcal{G}_{d_i-1}) + R(\mathcal{F}_{d_i}))}} \nonumber \\
&\le  &\expv{A_i \sum_{r = 1}^{\cardfdi} \expv{\frac{\ind{P_i \le \alpha_i(R(\mathcal{G}_{d_i-1}) + r), R(\mathcal{F}_{d_i}) = r}}{\alpha_i(R(\mathcal{G}_{d_i-1}) + r)} \cond R(\mathcal{G}_{d_i-1}), A_i}} \nonumber \\
&\le  &\expv{A_i \sum_{r = 1}^{\cardfdi} \pr{R(\mathcal{F}_{d_i}) = r \cond P_i \le \alpha_i(R(\mathcal{G}_{d_i-1}) + r), R(\mathcal{G}_{d_i-1}), A_i}} \nonumber \\
&\le  &\expv{A_i \pr{R(\mathcal{F}_{d_i}) \ge 1 \cond P_i \le \alpha_i(R(\mathcal{G}_{d_i-1}) + 1), R(\mathcal{G}_{d_i-1}), A_i}} \nonumber \\
&\le  &\expv{A_i}.
\end{eqnarray*}
The first equality follows by the fact that the event $\{H_i \text{ is rejected}\}$ is equivalent to the event $\{H_{T(i)} \text{ is rejected}, P_i \le \alpha_i(\Rd)\}$.  The second inequality follows by the fact that $P_i$ is independent of $R(\mathcal{G}_{d_i-1})$ and $A_i$ due to Assumption \ref{ASM_BLOCK_DEPENDENCE} so that (\ref{EQN_UNIFORM}) still holds. Finally, the third inequality is due to Lemma \ref{LEMMA_POSITIVE_DEPENDENCE}.

\vspace{10pt}
Now, we will show that
\begin{equation}
\expv{A_i \frac{V(\mathcal{M}_i)}{R}} \le \frac{\ell_i \alpha}{\ell}\expv{A_i}. \label{EQN_FDR_INDEP_1}
\end{equation}
\emph{Proof of (\ref{EQN_FDR_INDEP_1}).}
If $H_i$ is a false null leaf hypothesis, then the left hand side of (\ref{EQN_FDR_INDEP_1}) is 0. If $H_i$ is a true null leaf hypothesis, then
\begin{eqnarray*}
&& \expv{A_i \frac{V(\mathcal{M}_i)}{R}}
 \le  \expv{A_i \frac{\ind{H_i \text { is rejected}}}{\Rd}} \\
 & = &    \frac{\ell_i \alpha}{\ell}\expv{A_i \frac{\ind{H_i \text { is rejected}}}{\alpha_i(\Rd)}}
 \le  \frac{\ell_i \alpha}{\ell}\expv{A_i}.
\end{eqnarray*}
The first inequality follows by the fact that $\Rd \le R$ and $\V{\mathcal{M}_i} = \ind{H_i \text{ is rejected}}$ when $H_i$ is a true null leaf hypothesis. The equality follows by $\alpha_i(r) = r\alpha/\ell$ and $\ell_i = 1$. The second inequality follows by (\ref{EQN_FDR_INDEP_3}). Thus, (\ref{EQN_FDR_INDEP_1}) holds when $H_i$ is a leaf hypothesis.

Now, we will show that (\ref{EQN_FDR_INDEP_1}) holds when $H_i$ is a non-leaf hypothesis. By induction assume (\ref{EQN_FDR_INDEP_1}) holds for all children of $H_i$. If $H_i$ is false, then we note that $V(\mathcal{M}_i) = \csum{j:T(j) = i}{} V(\mathcal{M}_j)$ and $A_i = A_j$ for each $j$ such that $T(j) = i$. Thus,
\begin{eqnarray*}
& &     \expv{A_i \frac{V(\mathcal{M}_i)}{R}}
=    \expv{\sum_{j: T(j) = i}A_j\frac{V(\mathcal{M}_j)}{R}} \\
& \le &  \sum_{j: T(j) = i}\frac{\ell_j\alpha}{\ell}\expv{A_j}
=    \frac{\ell_i \alpha}{\ell}\expv{A_i}.
\end{eqnarray*}
The inequality follows by induction.

Now, assume $H_i$ is true. We will use the following inequality,
\begin{eqnarray}
 \frac{1}{R} & = &  \frac{1}{\Rd} - \frac{R-\Rd}{R \Rd} \nonumber \\
& \le & \frac{1}{\Rd} - \sum_{j:T(j) = i}\frac{R(\mathcal{M}_j)}{R \Rd} \nonumber \\
& \le & \frac{1}{\Rd} - \sum_{j:T(j) = i}\frac{V(\mathcal{M}_j)}{R \Rd}. \label{EQN_FDR_INDEP_4}
\end{eqnarray}
The equality follows by simple algebra and the first inequality follows by the fact that $\mathcal{M}_j \subseteq \mathcal{M} / \mathcal{G}_{d_i}$ for each $j$ with $T(j) = i$ so that $\sum_{j:T(j) = i}R(\mathcal{M}_j) \le R-\Rd$. The second inequality follows by the fact that
$R(\mathcal{M}_j) \ge V(\mathcal{M}_j)$ for each $j$. It should also be noted that $V(\mathcal{M}_i) = (1 + \sum_{j:T(j) = i}V(\mathcal{M}_j))\ind{H_i \text{ is rejected}}$. Thus,
\begin{eqnarray*}
   &&\expv{A_i\frac{V(\mathcal{M}_i)}{R}} \\
&=   &\expv{A_i\left(\frac{1}{R} + \sum_{j:T(j) = i}\frac{V(\mathcal{M}_j)}{R}\right)\ind{H_i \text{ is rejected}}} \\
&\le &\expv{A_i\left(\frac{1}{\Rd} - \sum_{j:T(j) = i}\frac{V(\mathcal{M}_j)}{R\Rd} + \sum_{j:T(j) = i}\frac{V(\mathcal{M}_j)}{R}\right)\ind{H_i \text{ is rejected}}} \\
&=   &\expv{A_i\left(\frac{1}{\Rd} + \left(1 - \frac{1}{\Rd}\right)\sum_{j:T(j) = i}\frac{V(\mathcal{M}_j)}{R}\right)\ind{H_i \text{ is rejected}}} \\
&=   &\expv{A_i\frac{\ind{H_i \text{ is rejected}}}{\Rd} + \sum_{j:T(j) = i}A_j\frac{V(\mathcal{M}_j)}{R}} \\
&\le &\expv{A_i\frac{\ind{H_i \text{ is rejected}}}{\Rd} + \sum_{j:T(j) = i}\frac{\ell_j \alpha}{\ell}A_j} \\
&=   &\expv{A_i\left(\frac{1}{\Rd} + \left(1 - \frac{1}{\Rd}\right)\frac{\ell_i \alpha}{\ell}\right)\ind{H_i \text{ is rejected}}} \\
&=   &\frac{\ell_i\alpha}{\ell}\expv{A_i\frac{\ind{H_i \text{ is rejected}}}{\alpha_i(\Rd)}} \\
&\le &\frac{\ell_i\alpha}{\ell}\expv{A_i}.
\end{eqnarray*}
The first inequality follows by (\ref{EQN_FDR_INDEP_4}). The third and forth equality follow by the fact that $A_j = A_i(1 - (1/\Rd)\ind{H_i \text{ is rejected}}$ for $j$ such that $T(j) = i$, since $H_i$ is true. The second inequality follows by induction. The last equality follows by $\alpha_i(r) = \ell_i r \alpha/(\ell+\ell_i(r-1)\alpha)$ and the last inequality follows by (\ref{EQN_FDR_INDEP_3}).

\vspace{10pt}
\noindent\emph{Proof of Theorem \ref{THM_HIER_FDR_BLOCK_POS}.} Finally, by (\ref{EQN_FDR_INDEP_1}),
\begin{eqnarray*}
& & \text{FDR} = \expv{\frac{V}{R}} = \sum_{i:T(i) = 0}\expv{\frac{V(\mathcal{M}_i)}{R}} \\
&=& \sum_{i:T(i) = 0}\expv{A_i\frac{V(\mathcal{M}_i)}{R}} \le \sum_{i:T(i) = 0}\frac{\ell_i \alpha}{\ell} = \alpha. ~\qed
\end{eqnarray*}

\subsection{Proof of Theorem \ref{THM_HIER_FDR_BLOCK_ARB}}
In the proof, we will use the same notations as in the proof of Theorem \ref{THM_HIER_FDR_BLOCK_POS}. To prove Theorem \ref{THM_HIER_FDR_BLOCK_ARB}, we show that the following inequality holds when $H_i$ is true
\begin{equation}
\expv{A_i \frac{\ind{H_i \text{ is rejected}}}{\alpha_i(\Rd)}} \le c_i\expv{A_i}. \label{EQN_FDR_INDEP_5}
\end{equation}

\emph{Proof of (\ref{EQN_FDR_INDEP_5}).} It can be easily shown through simple algebra that for any leaf or non-leaf hypothesis $H_i$, the constant $c_i$ can be expressed as
\begin{equation}\label{EQN_Const}
c_i = 1 + \sum_{r=2}^{\cardfdi}\frac{\alpha_i(r+d_i-1) - \alpha_i(r+d_i-2)}{\alpha_i(r+d_i-1)}.
\end{equation}
Assume $H_i$ is true. Then,
\begin{eqnarray*}
   &  &\expv{A_i \frac{\ind{H_i \text{ is rejected}}}{\alpha_i(\Rd)}} \\
&=    &\expv{A_i \frac{\ind{P_i \le \alpha_i(R(\mathcal{G}_{d_{i}})), H_{T(i)} \text{ is rejected}}}{\alpha_i(R(\mathcal{G}_{d_i}))}} \nonumber \\
&=    &\expv{A_i \frac{\ind{P_i \le \alpha_i(R(\mathcal{G}_{d_i-1}) + R(\mathcal{F}_{d_i})), H_{T(i)} \text{ is rejected}}}{\alpha_i(R(\mathcal{G}_{d_i-1}) + R(\mathcal{F}_{d_i}))}} \nonumber \\
&=    &E\Bigg(A_i\ind{H_{T(i)} \text{ is rejected}} \times  \\
&     &\hspace{25pt} \left.\expv{\sum_{r = 1}^{\cardfdi}\frac{\ind{P_i \le \alpha_i(R(\mathcal{G}_{d_i-1}) + r), R(\mathcal{F}_{d_i}) = r}}{\alpha_i(R(\mathcal{G}_{d_i-1}) + r)} \cond \mathbf{P}_{d_i -1}}\right) \nonumber \\
&=    &E\Bigg(A_i\ind{H_{T(i)} \text{ is rejected}} \times  \\
&    &\hspace{25pt}\left.\sum_{r = 1}^{\cardfdi}\frac{\pr{P_i \le \alpha_i(R(\mathcal{G}_{d_i-1}) + r), R(\mathcal{F}_{d_i}) = r \cond \mathbf{P}_{d_i -1}}}{\alpha_i(R(\mathcal{G}_{d_i-1}) + r)}\right) \nonumber \\
&\le  &\expv{A_i \ind{H_{T(i)} \text{ is rejected}} \left(1 + \sum_{r = 2}^{\cardfdi}\frac{\alpha_i(R(\mathcal{G}_{d_i-1}) + r) - \alpha_i(R(\mathcal{G}_{d_i-1}) + r-1)}{\alpha_i(R(\mathcal{G}_{d_i-1}) + r)}\right)} \nonumber \\
&\le  &\expv{A_i \ind{R(\mathcal{G}_{d_i-1}) \ge d_i-1} \left(1 + \sum_{r = 2}^{\cardfdi}\frac{\alpha_i(R(\mathcal{G}_{d_i-1}) + r) - \alpha_i(R(\mathcal{G}_{d_i-1}) + r-1)}{\alpha_i(R(\mathcal{G}_{d_i-1}) + r)}\right)} \nonumber \\
&\le  &\expv{A_i \left(1 + \sum_{r = 2}^{\cardfdi}\frac{\alpha_i(r + d_i -1) - \alpha_i(r + d_i - 2)}{\alpha_i(r + d_i - 1)}\right)} \nonumber \\
&=    &c_i \expv{A_i}.
\end{eqnarray*}
Here, $\mathbf{P}_{d_i -1}$ denotes the $p$-value vector consisting of the $p$-values corresponding to the hypotheses in the first $d_i-1$ families, $\mathcal{F}_1, \dots, \mathcal{F}_{d_i-1}$.
The first equality follows by the fact that the event $\{H_i \text{ is rejected}\}$ is equivalent to the event  $\{H_{T(i)} \text{ is rejected}, P_i \le \alpha_i(\Rd)\}$. The first inequality follows by Lemma \ref{LEMMA_ARBITRARY_DEPENDENCE} and the fact that $P_i$ is independent of $\mathbf{P}_{d_i -1}$ due to Assumption \ref{ASM_BLOCK_DEPENDENCE} and $R(\mathcal{G}_{d_i-1})$ is determined by
$\mathbf{P}_{d_i -1}$. The second inequality follows by the fact that the event $\{H_{T(i)} \text{ is rejected}\}$ implies all ancestors of $H_{T(i)}$ are rejected, so there must be at least $d_i-1$ rejections in the first $d_i-1$ families, i.e., $R(\mathcal{G}_{d_i-1}) \ge d_i - 1$. The third inequality follows by the fact that $[\alpha_i(R(\mathcal{G}_{d_i-1}) + r) - \alpha_i(R(\mathcal{G}_{d_i-1}) + r-1)]/\alpha_i(R(\mathcal{G}_{d_i-1}) + r)$ is a decreasing function of $R(\mathcal{G}_{d_i-1})$ for each given $r$. The last equality follows from (\ref{EQN_Const}).

\vspace{10pt}
\noindent\emph{Proof of Theorem \ref{THM_HIER_FDR_BLOCK_ARB}.} By using the same argument for the proof of (\ref{EQN_FDR_INDEP_1}), we have that
$$\expv{A_i V(\mathcal{M}_i)/R} \le \ell_i\alpha\expv{A_i}/\ell.$$
Thus, the FDR control of this procedure follows by the same argument used in the proof of Theorem \ref{THM_HIER_FDR_BLOCK_POS}. \qed

\def\bibfont{\small}
\bibliographystyle{ECA_jasa}
\bibliography{references}

\end{document}